\documentclass[a4paper,twoside]{amsart}
\usepackage{amsmath}
\usepackage{amsfonts}
\usepackage{amssymb}
\usepackage{amsthm}
\usepackage{newlfont}
\usepackage{graphicx}
\usepackage{amscd}

\theoremstyle{plain}
\newtheorem{Th}{Theorem}[section]
\newtheorem{Cor}[Th]{Corollary}
\newtheorem{Lem}[Th]{Lemma}
\newtheorem{Prop}[Th]{Proposition}
\theoremstyle{definition}

\newtheorem{Ex}{Example}[section]
\theoremstyle{remark}
\newtheorem*{Rem}{Remark}
\numberwithin{equation}{section}

\newcommand{\NN}{{\mathbb N}}

\newcommand{\CC}{{\mathbb C}}
\newcommand{\DD}{{\mathbb D}}

\newcommand{\II}{{\mathbb I}}
\newcommand{\ZZ}{{\mathbb Z}}

\newcommand{\RR}{{\mathbb R}}

\begin{document}

\title[Spectral quantization of random walks, and orthogonal polynomials]
{Spectral quantization of discrete random walks on half-line, \\and orthogonal polynomials on the unit circle}

\author{Adam Doliwa and Artur Siemaszko}

\address{Faculty of Mathematics and Computer Science\\
University of Warmia and Mazury in Olsztyn\\
ul.~S{\l}oneczna~54\\ 10-710~Olsztyn\\ Poland} 
\email{doliwa@matman.uwm.edu.pl, artur@uwm.edu.pl}

%

\begin{abstract}
We define quantization scheme for discrete-time random walks on the half-line consistent with Szegedy's quantization of finite Markov chains. Motivated by the Karlin and McGregor description of discrete-time random walks in terms of polynomials orthogonal with respect to a measure with support in the segment $[-1,1]$, we represent the unitary evolution operator of the quantum walk in terms of orthogonal polynomials on the unit circle. We find the relation between transition probabilities of the random walk with the Verblunsky coefficients of the corresponding polynomials of the quantum walk. We show that the both polynomials systems and their measures are connected by the classical  Szeg\H{o} map. Our scheme can be applied to arbitrary Karlin and McGregor random walks and  generalizes the so called Cantero--Gr\"{u}nbaum--Moral--Vel\'{a}zquez method. 
We illustrate our approach on example of random walks related to the Jacobi polynomials.
Then we study quantization of random walks with constant transition probabilities where the corresponding polynomials on the unit circle have two-periodic real Verblunsky coefficients. We present geometric construction of the spectrum of such polynomials (in the general complex case) which generalizes the known construction for the Geronimus polynomials. In the Appendix we present the explicit form, in terms of Chebyshev polynomials of the second kind, of polynomials orthogonal on the unit circle and polynomials orthogonal on the real line with coefficients of arbitrary period. 

\end{abstract}
\maketitle

\section{Introduction}

Randomness and probabilistic methods play an important role in modern computer science~\cite{MitzenmacherUpfal}. In quantum computing~\cite{Hirvensalo,NielsenChuang} the  measurement-induced randomness is inherent and, together with entanglement and interference, provides the three pillars of quantum algorithms. The famous Grover's search algorithm~\cite{Grover}, based on the amplitude amplification, can be put within the quantum walk on a finite graph, i.e. the quantum analogue of a Markov chain~\cite{ShenviKempeWhaley,Szegedy,Portugal}. Recent works show the fundamental role of quantum walks in the domain of quantum algorithms~\cite{1sQEW,Childs,Kempe,Kendon,MNRS,QW-review}. 

Quantum walks come in two variants: discrete and continuous time, and both have been shown to constitute a universal model of quantum computation~\cite{Childs,LCETK}. In this work we focus on the discrete-time model with coins, which consists of a bipartite quantum state, containing information about the position and coin states of a walker, and a unitary evolution operator, composed of the subsequent application of a coin operator, and a shift operator. A general quantization procedure for discrete-time finite Markov chains has been proposed by Szegedy~\cite{Szegedy}. One can associate eigenvalues of the quantum evolution operator and the eigenvalues of a self-adjoint matrix, called the discriminant operator, constructed from the stochastic matrix of an initial Markov process.

Orthogonal polynomials on real line (OPRL)~\cite{Chihara,Geronimus,Ismail,Szego} provide an important tool to investigate stochastic processes \cite{KarlinMcGregor,VanDoornSchrijner,Schoutens,dlIglesia}.The corresponding approach to quantum walks was initiated in~\cite{CMGV-1,CMGV-2} and is based on the theory of orthogonal polynomials on the unit circle (OPUC) \cite{Szego,Geronimus,Golinskii,Simon-OPUC}. Theory of OPUC exploits the canonical form of a unitary operator with the so called Cantero--Moral--Vel\'{a}zquez (CMV) matrices~\cite{CMV-1,CMV-2,Simon-CMV}, which for OPUC play the same role as the Jacobi matrices for OPRL.

The aim of the present work is to provide quantum analogs of arbitrary Karlin and McGregor random walks, thus generalizing the Cantero--Gr\"{u}nbaum--Moral--Vel\'{a}zquez (CGMV)  method~\cite{CMGV-1,CMGV-2}. Our scheme, which we call spectral quantization, provides a meaning to the classical Szeg\H{o} map~\cite{Szego} between OPUC with real Verblunsky coefficients and OPRL with the measure concentrated within the segment $[-1,1]$ of the real line. In our work we give clear definition of quantization of a random walk with arbitrary probabilities, what enables to  compare both processes.

Our paper is organized as follows.
In Section~\ref{sec:preliminaries} we recall basic theory of OPUC, their relation with CMV matrices, the Karlin and McGregor representation of discrete-time random walks, and Szegedy's quantization of finite Markov chains. Section~\ref{sec:Sz-CMV-big} is devoted to application of Szegedy's quantization scheme for random walks on half-line in conjunction with realization of unitary operators in terms of CMV matrices. In particular we show the place of the CGMV method within our approach. Then in Section~\ref{sec:Sz-CMV-spec} we discuss the quantum---classical correspondence on the level of the corresponding orthogonal polynomials and their measures. It turns out that our scheme recovers the classical Szeg\"{o} map between OPRL and OPUC. We illustrate the approach on example of the Jacobi polynomials. We remark also on another ``natural'' way to obtain discrete-time random walks from quantum walks. We devoted Section~\ref{sec:constant-prob} to detailed study of the quantization of discrete-time random walks with constant transition probabilities, what results in OPUC with two-periodic Verblunsky coefficients. In particular we show the explicit forms of both families of polynomials in terms of Chebyshev polynomials of the second kind. We close the section with geometric construction of the spectrum of OPUC with two-periodic complex Verblunsky coefficients, which for constant coefficients reduces to the known such construction for the Geronimus polynomials. Finally, motivated by the findings of the last section we give in the Appendix the explicit form of OPUC and OPRL with coefficients of arbitrary period.

\section{Preliminaries} \label{sec:preliminaries}

\subsection{Orthogonal polynomials on the unit circle}
Let $\DD = \{z \colon \, |z|<1 \} \subset \CC$ be the open unit disk, and let $\mu$  be a measure on  the unit circle  $\partial \DD = \{z \colon \, |z|=1 \}$ parametrized by $z=e^{i\theta}$. Unless stated explicitly, throughout the paper we assume that $\mu$ is nontrivial (i.e., supported on an infinite set) probability measure (i.e., $\mu$ is nonnegative and normalized by $\mu(\partial\DD) = 1)$. 

In the Hilbert space $\mathcal{H} = L^2(\partial\DD,d\mu) $ with the inner product antilinear in the left factor, we define the monic polynomials $\Phi_n(z)$, $n=0,1,2,\dots$ by the Gram-Schmidt orthogonalization procedure of the standard basis $1,z,z^2,\dots $. We have then
\begin{equation*}
\langle \Phi_n , \Phi_m \rangle = \frac{1}{\kappa_n^2} \delta_{nm}, \qquad \kappa_n > 0,
\end{equation*}
and the orthonormal polynomials $\varphi_n = {\kappa_n} \Phi_n$ satisfy
$\langle \varphi_n , \varphi_m \rangle = \delta_{nm}$.

If $P_n$ is a polynomial of degree $n$, define $P^*_n$, the reversed polynomial, by 
\begin{equation*}
P_n^*(z) = z^n \overline{P_n(1/\bar{z})}, \quad \text{i.e.} \quad 
P_n(z) = \sum_{j=0}^n c_j z^j \; \Rightarrow \; P_n^*(z) = \sum_{j=0}^n \bar{c}_{n-j} z^j .
\end{equation*}
The orthogonal polynomials $\Phi_n$ are given by the Szeg\H{o} recurrence
\begin{equation*} \Phi_0(z) = 1, \qquad 
\Phi_{n+1}(z) = z\Phi_n(z) - \bar\alpha_n \Phi^*_n(z), \qquad \alpha_n = - \overline{\Phi_{n+1}(0)}, \qquad n\geq 0, 
\end{equation*}
where the Verblunsky coefficients $\alpha_0, \alpha_1, \alpha_2,\dots$ satisfy $|\alpha_j| < 1$. By Verblunsky's theorem the map $\mu \to \{\alpha_j\}_{j=1}^\infty$ sets-up a bijection between the set of nontrivial probability measures on $\partial\DD$ and $\times_{j=1}^\infty \DD$. The Szeg\H{o} recurrence relations for orthonormal polynomials are
\begin{equation} \label{eq:phi-rec}
\left( \begin{array}{c} \varphi_{n+1}(z)\\\varphi_{n+1}^*(z) \end{array} \right) = \frac{1}{\rho_n} 
\left( \begin{array}{cc} z & -\bar{\alpha}_n \\ -\alpha_n z &1 \end{array} \right)
\left( \begin{array}{c} \varphi_{n}(z)\\\varphi_{n}^*(z) \end{array} \right) = A(\alpha_n) 
\left( \begin{array}{c} \varphi_{n}(z)\\\varphi_{n}^*(z) \end{array} \right),
\end{equation}
\begin{equation*} \rho_n = \sqrt{1 - |\alpha_n|^2}, \qquad \varphi_0(z) = \varphi_0^*(z) \equiv 1.
\end{equation*}

The Carath\'{e}odory function
\begin{equation*}
F(z) = \int_{\partial\DD} \, \frac{e^{i\theta} + z}{e^{i\theta}-z} \, d\mu(\theta)
\end{equation*}
is analytic on $\DD$ with non-negative real part, and normalized by $F(0) = 1$. If
\begin{equation*}
d\mu (\theta) = w(\theta)\frac{d\theta}{2\pi} + d\mu_s(\theta)
\end{equation*}
is decomposition of the measure into absolutely continuous and singular parts then:
\begin{enumerate}
	\item The weight is given by $w(\theta) = \lim_{r\to 1^-} \mathrm{Re} F(re^{i\theta})$.
	\item The singular part $\mu_s$ is concentrated on the set $\{ e^{i\theta} \colon \lim_{r\to 1^-} \mathrm{Re} F(re^{i\theta}) = \infty \}$.
	\item The mass of any point is given by $\mu(\{ \theta \}) = \lim_{r\to 1^-}\frac{1-r}{2} F(r e^{i\theta})$.
\end{enumerate}
The Schur function
\begin{equation*}
f(z) = \frac{1}{z} \frac{F(z) - 1}{F(z) + 1}
\end{equation*}
is analytic on $\DD$ with $|f(z)|<1$ for $z\in \DD$. By Geronimus' theorem $f$ has the following continued fraction decomposition
\begin{equation*}
f(z) = \alpha_0 + \cfrac{\rho_0^2 z}{\bar{\alpha}_0 z +
\cfrac{1}{\alpha_1 + \cfrac{\rho_1^2z}{\bar{\alpha}_1 z + \cfrac{1}{\cdots}}}} .
\end{equation*}

\subsection{Cantero--Moral--Vel\'{a}zquez matrices}
One of the key tools in the case of orthogonal polynomials on the real line is the realization of the measure as the spectral measure of the Jacobi matrix, which comes in as a matrix of multiplication by the real variable $x$. In the case of OPUC the corresponding matrix realization of the measure comes in terms of the CMV  martices.
\begin{Rem}
	These matrices appeared earlier in the numerical analysis literature~\cite{Watkins}.
\end{Rem}

Define the CMV basis $\{\chi_n\}_{n=1}^\infty$  by orthonormalizing the sequence $1,z,z^{-1}, z^2, z^{-2}, \dots$, and define matrix $\mathcal{C}$ by
\begin{equation*}
\mathcal{C}_{mn} = \langle \chi_m , z\chi_n \rangle.
\end{equation*}
The matrix is unitary and pentadiagonal 
\begin{equation*}
\mathcal{C} = \left( \begin{array}{ccccccc}
\bar{\alpha}_0 & \bar{\alpha}_1 \rho_0 & \rho_1\rho_0 & 0    & 0& 0 & \cdots\\
\rho_0  & -\bar{\alpha}_1 \alpha_0 & -\rho_1 \alpha_0 & 0& 0& 0 & \cdots \\
0 & \bar{\alpha}_2 \rho_1 & -\bar{\alpha}_2\alpha_1 & \bar{\alpha}_3 \rho_2 & \rho_3 \rho_2 &0& \cdots \\
0 & \rho_2 \rho_1 & -\rho_2\alpha_1 & -\bar{\alpha}_3 \alpha_2 & - \rho_3 \alpha_2 & 0&\cdots \\
0 & 0 & 0 & \bar{\alpha}_4 \rho_3 & -\bar{\alpha}_4\alpha_3 & \bar{\alpha}_5 \rho_4 & \cdots \\
0 & 0 & 0 & \rho_4 \rho_3 & -\rho_4\alpha_3 & -\bar{\alpha}_5 \alpha_4 & \cdots  \\
\cdots & \cdots & \cdots & \cdots & \cdots & \cdots & \cdots
\end{array} \right),
\end{equation*}
and has decomposition $\mathcal{C} = \mathcal{L}\mathcal{M}$, where
\begin{equation*}
\mathcal{L} = \Theta_0 \oplus \Theta_2 \oplus \Theta_4 \oplus \dots \, , \quad
\mathcal{M} = 1 \oplus \Theta_3 \oplus \Theta_ 5\oplus \dots \, , \qquad \text{and} \quad
\Theta_k = \left( \begin{array}{cc}
\bar{\alpha}_k & \rho_k \\ \rho_k & - \alpha_k
\end{array} \right) .
\end{equation*}
The monic orthogonal polynomials associated to the measure $\mu$ used to define $\mathcal{C}$ can be found by
\begin{equation*}
\Phi_n(z) = \det (z\mathbb{I}_n - \mathcal{C}^{(n)}),
\end{equation*}
where $\mathcal{C}^{(n)}$ is restriction of $\mathcal{C}$ to the upper $n\times n$ block.
The CMV basis can be expressed in terms of $\varphi$ and $\varphi^*$ by
\begin{equation} 
\label{eq:chi-phi}
\chi_{2k}(z) = z^{-k}\varphi_{2k}^*(z), \qquad \chi_{2k+1}(z) = z^{-k} \varphi_{2k+1}(z).
\end{equation}
\begin{Rem}
	In order to have formulas consistent it is custom to define $\alpha_{-1} = -1$.
\end{Rem}
When all $\alpha_k \in \DD$ then the vector $e_0 = (1,0,0,0, \dots)^T$ is cyclic for $\mathcal{C}$ in $\ell^2(\NN)$. In order to formulate the analog of Stone's self-adjoint cyclic model theorem recall that a cyclic unitary model is a unitary operator $U$ on a separable Hilbert space $\mathcal{H}$ with a distinguished unit vector $v_0$ such that finite linear combinations of $\{ U^n v_0 \}_{n\in\ZZ}$ are dense in $\mathcal{H}$. Two cyclic unitary models $(\mathcal{H},U,v_0)$ and $(\tilde{\mathcal{H}}, \tilde{U}, \tilde{v}_0)$ are called equivalent if there is unitary $W$ from $\mathcal{H}$ onto $\bar{\mathcal{H}}$ such that
\begin{equation*}
W U W^{-1} = \tilde{U}, \qquad W v_0 = \tilde{v}_0.
\end{equation*}
It turns out that each cyclic unitary model is equivalent to a unique CMV model $(\ell^2(\NN),\mathcal{C},e_0)$.
\begin{Rem}
	Finite $n\times n$ CMV matrices are characterized by $\alpha_k\in\DD$, $k=0,\dots ,n-2$, and $\alpha_{n-1}\in\partial\DD$. Each unitary equivalence class of $n\times n$ unitary matrices with a cyclic vector contains a unique finite CMV matrix with cyclic vector $e_0 = (1,0,\dots,0)^T$. 
\end{Rem}

\subsection{Karlin and McGregor approach to discrete random walks}
A one-dimensional random walk is a Markov chain in which the particle, if it is in state $k$, $k=0,1,2,\dots$, can in single transition either stay in $k$ or move to one of the adjacent states $k\pm 1$.  The one-step transition probabilities are given by the stochastic matrix
\begin{equation*}
\mathcal{P} = \left( \begin{array}{ccccc}
r_0 & p_0 & 0 & 0 & \cdots \\
q_1 & r_1 & p_1 & 0 & \cdots \\
0 & q_2 & r_2 & p_2 & \cdots \\
0 & 0 & q_3 & r_3 & \cdots \\
\vdots & \vdots & \vdots & \vdots & \ddots
\end{array} \right),
\end{equation*}
where we assume $p_k >0$, $q_{k+1}>0$, $r_k\geq 0$ for $k\geq 0$. The definition of the random walk implies that $p_k + r_k + q_k = 1$ for $k\geq 1$ and $p_0 + r_0 = 1$.
\begin{figure}[!ht]
	\begin{center}
		\includegraphics[width=7cm]{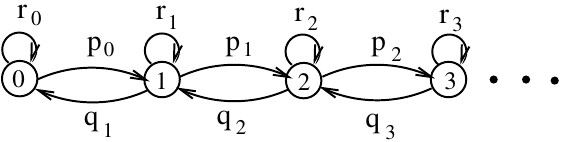}
	\end{center}
	\caption{Discrete random walk on the half-line}
	\label{fig:random-walk}
\end{figure}

Define constants $\pi_k$, $k\geq 0$, by
\begin{equation*}
\pi_0 = 1, \qquad \pi_k = \frac{p_0 p_1 \dots p_{k-1}}{q_1 q_2 \dots q_k}, \quad k\geq 1,
\end{equation*}
and find polynomials $\{Q_k(x)\}_{k=0}^\infty$ by the recurrence relations
\begin{align} \label{eq:3trr}
xQ_k(x) & = q_k Q_{k-1}(x) + r_k Q_k(x) + p_k Q_{k+1}(x), \qquad k\geq 1,\\
Q_0(x) & = 1, \qquad p_0Q_1(x) = x-r_0.
\end{align}
The polynomials are orthogonal with respect to a positive measure in the interval $[-1,1]$ of total mass~$1$ and infinite support, denoted by $\nu$.
The Cauchy--Stieltjes transform of the measure
\begin{equation*}
S(z) = \int_{-1}^1 \frac{d\nu(x)}{x-z} , \qquad z\in\CC\setminus [-1,1],
\end{equation*}
has the following expansion into continued fractions
\begin{equation*}
S(z) = \cfrac{-1}{z-r_0 + \cfrac{p_0 q_1}{z - r_1 + \cfrac{p_1 q_2}{\dots}}}, \qquad \lim_{z\to\infty}S(z) = 0.
\end{equation*}
The $n$-step transition probabilities $P_{ij}(n)$ from state $i$ to state $j$ may be represented as
\begin{equation*}
P_{ij}(n) = \pi_j \int_{-1}^1 x^n Q_i(x) Q_j(x) \, d\nu(x).
\end{equation*}

If the measure is decomposed into absolutely continuous and discrete parts  
\begin{equation*}
 d\nu(x) = u(x)dx + d\nu_s(x),
\end{equation*}
\begin{enumerate}
	\item The weight is given by $u(x) = \lim_{\varepsilon\to 0^+}\frac{1}{\pi} \mathrm{Im}S(x+i\varepsilon)$.
	\item The singular part $\nu_s$ is concentrated on the set $\{ x_0 \colon \lim_{x\to x_0} \mathrm{Im} S(x) = \infty \}$.
	\item The mass of any point is given by $\nu(\{ x_0 \}) = \lim_{\varepsilon\to 0^+} \varepsilon \mathrm{Im} S(x_0 + i\varepsilon )$.
\end{enumerate}
\begin{Rem}
The spectral measure originates from the Jacobi symmetric matrix
\begin{equation} \label{eq:J-CMG}
\mathcal{J} = \left( \begin{array}{ccccc}
r_0 & \sqrt{p_0 q_1} & 0 & 0 & \cdots \\
\sqrt{p_0 q_1} & r_1 & \sqrt{p_1 q_2} & 0 & \cdots \\
0 & \sqrt{p_1 q_2} & r_2 & \sqrt{p_2 q_3} & \cdots \\
0 & 0 & \sqrt{p_1 q_2} & r_3 & \cdots \\
\vdots & \vdots & \vdots & \vdots & \ddots
\end{array} \right),
\end{equation}
which represents the same linear operator but in new basis obtained by suitable scaling of vectors of the starting one.
\end{Rem}

\subsection{Szegedy's quantization of Markov chains}
A discrete-time classical random walk on a finite set of states $V$, where $|V|=N$, can be represented by an $N\times N$ stochastic matrix $P$, whose entry $P_{jk}$ represents the probability of making a transition from $j$ to $k$, in particular $\sum_{k=1}^N P_{jk} = 1$. Such a matrix gives rise to a directed weighted graph $G(V,E)$, where the ordered pair $(j,k)$ of vertices belongs to the edge set $E$ if and only if $P_{jk}>0$. The normalized row eigenvector $\pi\in \RR^N$ of $P$ corresponding to left eigenvalue~$1$ 
\begin{equation*}
\pi P = \pi, \qquad \sum_{k=1}^N \pi_k = 1,
\end{equation*}
represents a stationary state of the random walk.

Definition of Szegedy's quantum walk starts with tensor doubling $\CC^N \otimes \CC^N$ of the state space --- the state $|j \rangle \otimes |k\rangle = | j, k \rangle$ will be interpreted as "particle in position $j$ looks at the position $k$". The first factor is called the position space, and the second factor is called the coin space. The stochastic matrix $P$ allows to define normalized orthogonal vectors
\begin{equation} \label{eq:phi-i}
| \phi_j \rangle =  |j \rangle \otimes \sum_{k=1}^N \sqrt{P_{jk}}\,  |k\rangle = \sum_{k=1}^N \sqrt{P_{jk}}\,  |j, k\rangle. 
\end{equation} 
By $\Pi$ denote the orthogonal projection on the subspace generated by the vectors  $|\phi_j\rangle$, i.e. 
\begin{equation} \label{eq:Pi}
\Pi = \sum_{j=1}^N | \phi_j \rangle \langle \phi_j |,
\end{equation}
then the reflection
\begin{equation} \label{eq:R}
R = 2\Pi - \II,
\end{equation}
which acts on the coin space only, will be called the coin flip operator. Notice that it does not change the vector $|\phi_j\rangle$, which therefore can be considered as quantum analogue of a stationary coin state at the vertex $j$. 

Let $S$ be the operator that swaps the two registers 
\begin{equation} \label{eq:S}
S = \sum_{j,k=1}^N | j  , k \rangle \langle k , j | ,
\end{equation}
then the single step of the quantum walk is defined as the unitary operator being the composition of coin flip and the position swap
\begin{equation} \label{eq:U}
U = S \, R.
\end{equation}
\begin{Prop} \label{prop:Szegedy-Markov}
	The probability of finding the particle in position $k$ after one step of the quantum walk when starting from the state $| \phi_j \rangle$ is $P_{jk}$.
\end{Prop}
\begin{proof}
	The amplitude of finding the particle in the state $|k, \ell \rangle$ equals
	\begin{equation*}
	\langle k,\ell | U \phi_j\rangle
	 = \langle k,\ell | S \phi_j\rangle =
	 \sum_{i=1}^N \sqrt{P_{ji}} \langle k,\ell | i, j \rangle =
	 \sqrt{P_{jk}} \delta_{j\ell},	
	\end{equation*} 
which implies the statement
	\begin{equation*}
\sum_{\ell = 1}^N |\langle k,\ell | U \phi_j\rangle|^2
 = P_{jk}.	
\end{equation*} 
\end{proof}
\begin{Rem}
	Although the above result is very natural and seems to motivate the quantization scheme, we were not able to find it in the literature. 
\end{Rem}
\begin{Rem}
	The  two-step quantum walk operator 
	\begin{equation*}
	U^2 = [S(2\Pi - \II)] [S( 2\Pi - \II)] = [2S\Pi S - \II][2\Pi -\II],
	\end{equation*}
	can be described as composition of the first reflection $R$ (the coin flip) with the second reflection about the subspace generated by the vectors  $S|\phi_j\rangle$. Notice that the second reflection is the position flip with fixed coin. Such an evolution operator is closer to Szegedy's original definition of the quantum walk in the spirit ot Grover's search algorithm~\cite{MNRS}. In fact, in Szegedy's paper~\cite{Szegedy} the second step is made with another stochastic matrix.
\end{Rem}
\begin{figure}[h!]
	\begin{center}
		\includegraphics[width=7cm]{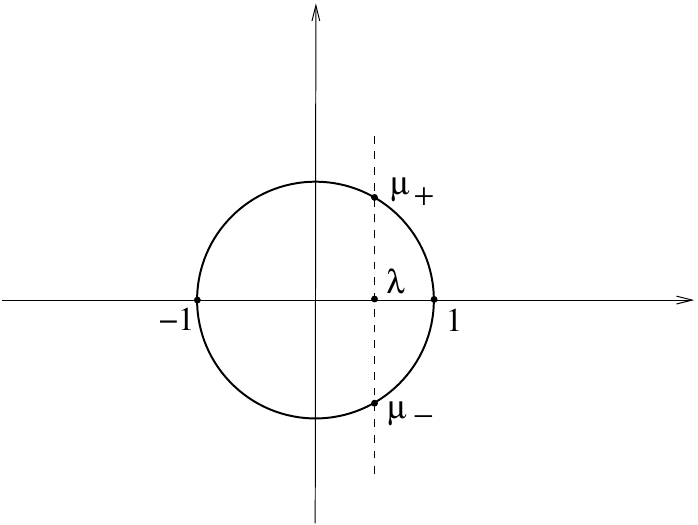}
	\end{center}
	\caption{Spectrum of the quantum walk operator in Szegedy's quantization}
	\label{fig:spect-Sz}
\end{figure}
The spectrum of the quantum walk operator $U$ can be computed directly from its matrix representation in the vertex-coin space. However in~\cite{Szegedy,Childs-CMP} it was shown how the spectrum of intermediate real symmetric matrix $D$, called the discriminant of $P$, with elements $D_{jk} = \sqrt{P_{jk} P_{kj}}$, which acts in the position space,  can be used to simplify the calculation. 
Let
\begin{equation*}
T = \sum_{j=1}^N |\phi_j \rangle \langle j |
\end{equation*}
be the isometry mapping $|j\rangle \in \CC^N$ to $| \phi_j \rangle \in \CC^N \otimes \CC^N$, then $T T^\dagger$ is the projection $\Pi$ onto $\mathrm{span} \{ |\phi_j\rangle \}$ and $T^\dagger T = \II_N$. 
\begin{Prop} \label{prop:spec-Sz}
When $\{ | \lambda \rangle \}$ denotes the complete set of eigenvectors of the $N\times N$ matrix \begin{equation*}
D=T^\dagger U T = T^\dagger S T = \sum_{j,k=1}^N \sqrt{P_{jk} P_{kj}} |j \rangle \langle k|,
\end{equation*}
with eigenvalues $\{ \lambda \}$, then the evolution operator $U$ has the corresponding eigenvectors 
\begin{equation*}
| \mu_\pm \rangle = T|\lambda\rangle - \mu_\pm ST |\lambda \rangle,
\end{equation*} 
with eigenvalues
\begin{equation*}
\mu_\pm = \lambda \pm i\sqrt{1 - \lambda^2} = e^{\pm i \arccos \lambda},
\end{equation*}
as visualized in Figure~\ref{fig:spect-Sz}.
The remaining eigenvalues of $U$ are $\pm 1$ with eigenvectors orthogonal to the subspace spanned by $T|\lambda \rangle$.
\end{Prop}

\section{Szegedy's quantization of random walks on half-line} \label{sec:Sz-CMV-big} 
\subsection{The direct application of Szegedy's scheme}
Let us apply Szegedy's quantization scheme to the discrete-time random walks on half-line. The coin space over vertex $k\geq 0$ is spanned by vectors
\begin{equation} \label{eq:vectors-coin}
| 0,0 \rangle, \; | 0, 1 \rangle  , \; \text{for} \; k=0, \; \text{and} \; 
|k, k-1 \rangle, \; |k, k \rangle, \; |k, k+1 \rangle ,  \; \text{for} \;  k >0.
\end{equation}
The corresponding distinguished states read
\begin{equation*}
|\phi_0 \rangle = \sqrt{r_0}\, |0,0\rangle + \sqrt{p_0} \,|0,1\rangle, \; \text{and} \quad 
|\phi_k \rangle = \sqrt{q_k} \,| k, k-1 \rangle + \sqrt{r_k}\, |k,k\rangle + \sqrt{p_k}\, |k,k+1\rangle, \quad k>0;
\end{equation*}
notice that the general formula works also for $k=0$, when taking into account that $q_0 = 0$.
With the lexicographic ordering of the vectors~\eqref{eq:vectors-coin} the quantum evolution operator $U=SR$ has the structure induced by the decompositions
\begin{equation*}
R = R_0 \oplus R_1 \oplus R_2 \oplus \dots
\end{equation*}
where
\begin{equation*}
R_0 = \begin{pmatrix}
2 r_0 - 1 & 2 \sqrt{p_0 r_0} \\
2 \sqrt{p_0 r_0} & 2 p_0 -1 
\end{pmatrix}, \quad \text{and} \quad
R_k = \begin{pmatrix}
2 q_k -1 & 2 \sqrt{q_k r_k} & 2 \sqrt{p_k q_k} \\
2 \sqrt{q_k r_k} & 2 r_k - 1 & 2 \sqrt{p_k r_k} \\
2 \sqrt{p_k q_k} & 2 \sqrt{p_k r_k} & 2 p_k -1 
\end{pmatrix}, \quad k > 0, 
\end{equation*}
and
\begin{equation*}
S = 1 \oplus A \oplus 1 \oplus A \oplus 1 \oplus \dots , \qquad 
A = \begin{pmatrix}
0 & 1 \\ 1 & 0 
\end{pmatrix},
\end{equation*}
as visualized on Figure~\ref{fig:S-walk}.
Define 
\begin{equation*}
|\psi_k \rangle = S|\phi_k \rangle = \sqrt{q_k} \,| k-1, k \rangle + \sqrt{r_k}\, |k,k\rangle + \sqrt{p_k}\, |k+1,k\rangle, \quad k\geq 0,
\end{equation*}
then the elements of the discriminant matrix are given by \begin{equation} \label{eq:D-CMV}
D_{jk} =\langle\phi_j | U \phi_k \rangle = \langle\phi_j | \psi_k \rangle,
\end{equation}
and the matrix  coincides with the Jacobi matrix $\mathcal{J}$~\eqref{eq:J-CMG} of Karlin and McGregor.
\begin{Rem}
	Due to unitarity of $S$ and the orthonormality  $\langle \phi_j | \phi_k \rangle = \delta_{jk}$  we have also $\langle \psi_j | \psi_k \rangle = \delta_{jk}$.
\end{Rem}
\begin{figure}[h!]
	\begin{center}
		\includegraphics[width=13cm]{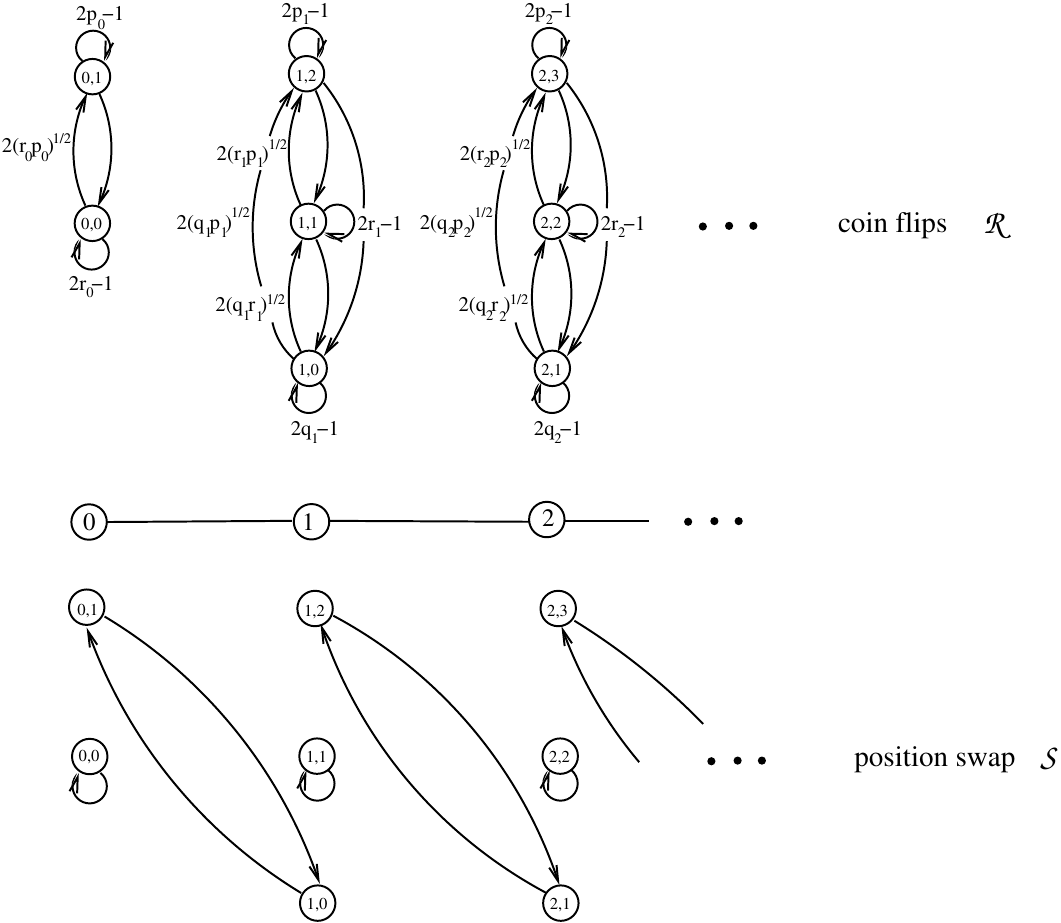}
	\end{center}
	\caption{Szegedy's quantization of the random walk on half-line}
	\label{fig:S-walk}
\end{figure}
\begin{Rem}
	Notice that Proposition~\ref{prop:Szegedy-Markov} applies also to the above quantization of discrete-time random walks.
\end{Rem}
\begin{Ex} \label{ex:r=0}
	When $r_k = 0$, $k=0,1,2,\dots$, then in Szegedy's quantization the vectors $|k,k\rangle$ are invariant (up to a trivial change of sign) with respect to the action of $U$ defined above. After removing them, in the restricted space define $U^\prime = S^\prime R^\prime$
	\begin{equation*}
	R^\prime = 1 \oplus R^\prime_1 \oplus R^\prime_2 \oplus \dots , \qquad R^\prime_k = \begin{pmatrix}
	2q_k - 1 & 2\sqrt{p_k q_k} \\ 2\sqrt{p_k q_k} & 2p_k -1 
	\end{pmatrix},
	\qquad S^\prime = A\oplus A\oplus A\oplus \dots \, .
	\end{equation*}
Then in the CMV picture $S^\prime$ can be identified with the matrix $\mathcal{L}$ for vanishing even Verblunsky coefficients $\alpha_{2k} = 0$, $k=0,1,2,\dots$, while $R^\prime$ plays the role of $\mathcal{M}$ with odd Verblunsky coefficients of the form
\begin{equation*}
\alpha_{2k-1} = q_k - p_k, \quad \text{or} \quad 
p_k = \frac{1}{2}(1-\alpha_{2k-1}), \; q_k = \frac{1}{2}(1+\alpha_{2k-1}) , \quad k = 0,1,2, \dots \,.
\end{equation*}	
The CMV basis is given by the vectors
\begin{equation*}
e_{2k} = |k, k+1\rangle, \qquad e_{2k+1} = | k+1, k \rangle, \qquad k = 0, 1, 2, \dots \,  .
\end{equation*}
See also \cite{CMGV-2} for more detailed discussion of this special reduction in relation to orthogonal polynomials.
\end{Ex}
\begin{Rem}
We would like to point out an important difference between the model of quantum walks studied by us and that considered in \cite{CMGV-1,CMGV-2}. Mainly, in the place of the \emph{swap operator} they use the \emph{shift operator}
	which in our notation reads 
\begin{equation*}
	| k, k\pm 1 \rangle \mapsto | \pm 1, k\pm 2 \rangle , \quad |k,k \rangle \mapsto |k,k\rangle .
\end{equation*}
\end{Rem}
\begin{Rem}
 It should be mentioned that in \cite{CMGV-1,CMGV-2} quantum walks on the full set of integers $\mathbb{Z}$ (with the shift operator) are considered, what leads to two natural extensions of CMV matrices: doubly infinite CMV matrices and block CMV matrices. Such a problem is beyond the scope of the present article.		
\end{Rem}

\subsection{CMV description of Szegedy's quantization of the random walks} \label{sec:Sz-CMV}
Let us go back to the generic case.
By $\mathcal{H}$ denote the Hilbert space with orthonormal basis \eqref{eq:vectors-coin} of the all coin vectors. We will be interested in its subspace $\mathcal{H}_0 = \mathrm{span}\{ | \phi_k \rangle, |\psi_k \rangle \}_{k\geq 0}$, whose importance follows from the following result.
\begin{Lem} For $k\in \NN$
	\begin{align*}
	U^k |\phi_0 \rangle = S(RS)^{k-1} |\phi_0 \rangle  & \in \mathrm{span} \{ |\phi_0 \rangle, |\psi_0 \rangle, \dots , |\phi_{k-1} \rangle,|\psi_{k-1} \rangle \} , \\
	U^{-k} |\phi_0 \rangle = (RS)^k |\phi_0 \rangle& \in \mathrm{span} \{ |\phi_0 \rangle, |\psi_0 \rangle, \dots , |\phi_{k-1} \rangle,|\psi_{k-1} \rangle, |\phi_{k} \rangle \}.
	\end{align*}
\end{Lem}
\begin{proof}
The first equalities are immediate consequence of the definition of $U$ and of the choice of $|\phi_0 \rangle$ as eigenvector of $R$. The second part will be demonstrated by induction. For $k=1$
\begin{equation*}
U|\phi_0 \rangle = S |\phi_0 \rangle = |\psi_0 \rangle ,\qquad
U^{-1} |\phi_0 \rangle = R S |\phi_0 \rangle = R |\psi_0 \rangle = 
2 |\phi_0 \rangle \langle \phi_0 | \psi_0 \rangle + 
2 |\phi_1 \rangle \langle \phi_1 | \psi_0 \rangle - |\psi_0 \rangle.
\end{equation*}
To show the induction step it is enough to consider the action $U^{\pm 1}$ on the last two vectors. We have
\begin{align*}
U|\phi_{k-1} \rangle & = S |\phi_{k-1} \rangle = |\psi_{k-1} \rangle ,\\
U |\psi_{k-1} \rangle & =S\left( 2 \sum_{j=0}^{k} |\phi_j \rangle \langle \phi_j | \psi_{k-1} \rangle - |\psi_{k-1} \rangle \right) =
2 \sum_{j=0}^{k} |\psi_j \rangle \langle \phi_j | \psi_{k-1} \rangle - |\phi_{k-1} \rangle
\end{align*}
and
\begin{align*}
U^{-1}|\psi_{k-1} \rangle & = R |\phi_{k-1} \rangle = |\phi_{k-1} \rangle ,\\
U^{-1} |\phi_{k} \rangle \; & = R | \psi_k \rangle = 
2 \sum_{j=0}^{k+1} |\phi_j \rangle \langle \phi_j | \psi_{k} \rangle - |\psi_{k} \rangle.
\end{align*}
\end{proof} 
\begin{Prop} \label{prop:pqr-a}
	The CMV basis of the quantum evolution operator for Szegedy's quantization of the random walk on the half-line with the cyclic vector $e_0 = | \phi_0 \rangle$ has the Verblunsky coefficients related with the random walk transition probabilities by the formulas
\begin{align}  \label{eq:qk}
	q_k & = \frac{1}{2} ( 1 + \alpha_{2k-2}) (1 + \alpha_{2k-1}), \\
	\label{eq:rk}
	r_k & = \frac{1}{2} \left( \alpha_{2k}(1-\alpha_{2k-1}) - \alpha_{2k-2}(1+\alpha_{2k-1})\right), \qquad k\geq 0,\\
	\label{eq:pk}
	p_k & = \frac{1}{2} ( 1 - \alpha_{2k-1}) (1 - \alpha_{2k}).
\end{align}
\end{Prop}
\begin{Rem}
Recall that we put $\alpha_{-1} = -1$ what gives $q_0 = 0$, $r_0 = \alpha_0$, $p_0 = 1- \alpha_0$.
\end{Rem}
\begin{proof}
	Define the orthonormal basis $(e_j)_{j\in\NN_0}$ and the coefficients $\alpha_j$ by the recursive application of the operators $S$ and $R$ on the initial vector $e_0 = | \phi_0 \rangle$:
	\begin{align} \label{eq:Se}
	S(e_{2k}) & = \alpha_{2k} e_{2k} + \rho_{2k} e_{2k+1}, \qquad k\geq 0, \\ \label{eq:Re}
	R(e_{2k+1}) & = \alpha_{2k+1} e_{2k+1} + \rho_{2k+1} e_{2k+2},
\end{align}
where the sign of $\rho_j = \sqrt{1-\alpha_j^2} >0$ fixes the orientation of $e_{j+1}$. 
Notice that 
\begin{equation*}
\alpha_{2k}=\langle e_{2k}|S e_{2k}\rangle, \qquad \alpha_{2k+1}=\langle e_{2k+1}|Re_{2k+1}\rangle.	
\end{equation*}
The condition $S^2 = R^2 = \II$ gives
	\begin{align*}
S(e_{2k+1}) & = \rho_{2k} e_{2k} - \alpha_{2k} e_{2k+1}, \\
R(e_{2k+2}) & = \rho_{2k+1} e_{2k+1} - \alpha_{2k+1} e_{2k+2},
\end{align*}
and inductively confirms that there are no vectors $e_j$ with smaller indices on right hand sides of formulas~\eqref{eq:Se} and \eqref{eq:Re}.

Because $|\phi_k \rangle$ is orthogonal to the subspace $\mathrm{span} \{ |\phi_0 \rangle, |\psi_0 \rangle, \dots , |\psi_{k-2} \rangle, |\phi_{k-1} \rangle \}$
then the decomposition
\begin{equation*}
|\phi_k \rangle = \gamma_{2k-1}e_{2k-1} + \gamma_{2k} e_{2k},
\end{equation*}
with the eigenvector condition $R|\phi_k \rangle = | \phi_k\rangle$ and the normalization $\langle \phi_k | \phi_k \rangle = 1$ allow to fix the coefficients up to a common sign and give
\begin{equation} \label{eq:phi-e}
|\phi_k \rangle = \pm \left( \sqrt{\frac{1+\alpha_{2k-1}}{2}} e_{2k-1} + \sqrt{\frac{1-\alpha_{2k-1}}{2}} e_{2k} \right).
\end{equation}
By application of $S$ on such $|\phi_k \rangle$
we get
\begin{equation} \label{eq:psi-e}
|\psi_k \rangle = \pm \left(  \sqrt{\frac{1+\alpha_{2k-1}}{2}} \left(\rho_{2k-2} e_{2k-2} - \alpha_{2k-2} e_{2k-1}\right) + \sqrt{\frac{1-\alpha_{2k-1}}{2}} \left(\alpha_{2k} e_{2k} + \rho_{2k} e_{2k+1}\right) \right).
\end{equation}
Then from $\langle \phi_k | \psi_k \rangle = r_k$, as shown in~\eqref{eq:D-CMV}, we directly obtain equation~\eqref{eq:rk}. Moreover the condition $\langle \phi_{k+1} | \psi_k \rangle = \sqrt{p_k q_{k+1}}$ implies that the signs in front of all coefficients $\gamma_j$ are the same and gives
\begin{equation} \label{eq:pq}
p_k q_{k+1} = \frac{1}{4} (1-\alpha_{2k-1}) (1 - \alpha_{2k}^2) (1+\alpha_{2k+1}).
\end{equation}
Finally, the initial conditions
\begin{equation*}
\gamma_0 = 1, \qquad q_0 = 0, \qquad  r_0 = \alpha_0, 
\end{equation*}
give positive signs of $\gamma_j$ (and then in equations~\eqref{eq:phi-e}, \eqref{eq:psi-e}), and allow to find from equations \eqref{eq:rk}, \eqref{eq:pq} and the probability normalization condition $p_k + r_k + q_k = 1$ recursively the expressions \eqref{eq:qk} and \eqref{eq:pk}. 
\end{proof}
\begin{Cor}
	In terms of the initial basis of Szegedy's quantization the first two elements of the CMV basis can be expressed as follows
\begin{align*}
e_0 & = \sqrt{\alpha_0} | 0, 0 \rangle + \sqrt{1-\alpha_0} |0,1 \rangle, \\
e_1 & = \frac{1}{\sqrt{1+\alpha_0}}\left( \sqrt{\alpha_0 (1-\alpha_0)} | 0, 0 \rangle  - \alpha_0 | 0, 1 \rangle + |1,0 \rangle \right).
\end{align*}
\end{Cor}
\begin{Rem}
	We would like to stress that different choices of the initial state (the cyclic vector) may result in changing the CMV space where the quantum evolution takes place. This does not apply for the states being already in $\mathcal{H}_0$, in particular to $|\phi_k \rangle$ or $|\psi_k \rangle$. 
\end{Rem}

\subsection{The orthogonal complement to the CMV subspace}
To conclude the discussion of the CMV description of Szegedy's quantization of the random walks with the cyclic vector $|\phi_0\rangle$ let us find the orthogonal complement in $\mathcal{H}$ of the subspace   $\mathcal{H}_0 = \mathrm{span}\{ | \phi_k \rangle, |\psi_k \rangle \}_{k\geq 0}$. In Example~\ref{ex:r=0} we considered the special case of $r_k = 0$ for all $k \in \NN_0$. To avoid degeneracies we assume in this section that $r_k > 0$ for all $k \in \NN_0$.

By direct calculation one can check the following result.
\begin{Lem}
	The vectors
	\begin{align*}
	| \tilde{\phi}_k \rangle & =\frac{1}{\sqrt{p_k + r_k}}\left( \sqrt{p_k} | k, k \rangle - \sqrt{r_k} | k, k+1 \rangle \right), \qquad k\geq 0, \\
	| \hat{\phi}_k \rangle & =\sqrt{p_k + r_k} | k, k-1\rangle - \sqrt{\frac{q_k}{p_k + r_k}}\left( \sqrt{r_k} | k, k \rangle + \sqrt{p_k} | k, k+1 \rangle \right), \qquad k\geq 1,
	\end{align*}
	are normalized orthogonal eigenvectors of the coin flip operator $R$ with eigenvalues $-1$. Together with $| \phi_k \rangle$ they form orthonormal basis over the vertex $k$. Moreover, the vectors $\{ | \tilde{\phi}_0 \rangle, | \tilde{\phi}_1 \rangle, | \hat{\phi}_1 \rangle, \dots \}$ form an orthonormal basis of the orthogonal complement to the subspace $\mathcal{H}_1 = \mathrm{span}\{ | {\phi}_k \rangle\}_{k\geq 0}$. 
\end{Lem}

\begin{Prop}
	The orthogonal complement in $\mathcal{H}$ of the subspace   $\mathcal{H}_0$ is spanned by the vectors
\begin{equation} \label{eq:sigma}
| \sigma_k \rangle = 
\sqrt{p_k r_{k+1}} |k, k\rangle - \sqrt{r_k r_{k+1}} \left( | k, k+1 \rangle + | k+1, k \rangle \right) + \sqrt{r_k q_{k+1}} |k+1, k+1 \rangle  , \qquad k\geq 0.
\end{equation}		
\end{Prop}
\begin{proof}
Consider vectors $| \Sigma_a \rangle \in \mathcal{H}_1^\perp$ which can be written as
	\begin{equation} \label{eq:Sigma}
	| \Sigma_a \rangle = \tilde{a}_0 |\tilde{\phi}_0 \rangle + 
	\tilde{a}_1 |\tilde{\phi}_1 \rangle + \hat{a}_1 |\hat{\phi}_1 \rangle +
	\tilde{a}_2 |\tilde{\phi}_2 \rangle + \hat{a}_2 |\hat{\phi}_2 \rangle + \dots \, ,
	\end{equation}
and consider closed subspaces defined by relations $\langle \psi_k | \Sigma_a \rangle=0$, $k=0,1,\dots$,
which in the coordinates read
	\begin{equation} \label{eq:a-rec-1}
	- \tilde{a}_{k-1} \sqrt{\frac{r_{k-1} q_k}{p_{k-1} + r_{k-1}}} -
	\hat{a}_{k-1} \sqrt{\frac{q_{k-1}p_{k-1}q_k}{p_{k-1} + r_{k-1}}} + 
	\tilde{a}_k \sqrt{\frac{p_k r_k}{p_k + r_k}} -
	\hat{a}_k r_k \sqrt{\frac{q_k}{p_k + r_k}} +
	\hat{a}_{k+1} \sqrt{p_k (p_{k+1} + r_{k+1})}  = 0.
	\end{equation}
Notice that equations for $k=0,1$  can be obtained from the above ones by the condition $q_0 = 0$, and read
\begin{gather} \label{eq:a-init-0}
	\tilde{a}_0 \sqrt{\frac{p_0 r_0}{p_0 + r_0}} + \hat{a}_1 \sqrt{p_0 (p_1 + r_1)}  = 0,\\ \label{eq:a-init-1}
- \tilde{a}_0 \sqrt{\frac{r_0 q_1}{p_0 + r_0}} + 
\tilde{a}_1 \sqrt{\frac{p_1 r_1}{p_1 + r_1}} -
\hat{a}_1 r_1 \sqrt{\frac{q_1}{p_1 + r_1}} +
\hat{a}_2 \sqrt{p_1 (p_2 + r_2)}  = 0.
\end{gather}
The subspace $\mathcal{H}_0^\perp$ is the intersection of all such hyperplanes. 

The above equations can be simplified, i.e. \eqref{eq:a-rec-1} and \eqref{eq:a-init-1} can be replaced by the following relation
\begin{equation*}
	\tilde{a}_k \sqrt{r_k} + \hat{a}_k \sqrt{p_k q_k} + \hat{a}_{k+1} \sqrt{(p_k + r_k)(p_{k+1} + r_{k+1})} = 0, \qquad k >0.
\end{equation*}
By taking $\hat{a}_k$, $k= 1,2,\dots$, as independent parameters we define vectors $| s_k\rangle$ with coefficients $\hat{a}_j = \delta_{jk}$
\begin{equation*}
	|s_k\rangle = | \hat{\phi}_k \rangle - \sqrt{ \frac{(p_{k-1} + r_{k-1})(p_k + r_k)}{r_{k-1}}} | \tilde{\phi}_{k-1} \rangle - \sqrt{\frac{p_k q_k}{r_k}} |\tilde{\phi}_k \rangle. 
\end{equation*}
The vectors
are linearly independent, and after eventual orthonormalization provide orthonormal basis of $\mathcal{H}_0^\perp$. In working with the natural basis of the full Hilbert space $\mathcal{H}$ it is convenient to rescale the vectors and define 
\begin{equation*}
	|\sigma_k \rangle = - \sqrt{\frac{r_k r_{k+1}}{p_{k+1} + r_{k+1}}} | s_{k+1}  \rangle, \qquad k = 0,1,2, \dots ,
\end{equation*}
what provides formulas~\eqref{eq:sigma}. 
\end{proof}
\begin{Rem}
	Vectors of the subspace $\mathcal{H}_0^\perp$ have the following approximation expansion
	\begin{equation*}
		| \Sigma_a \rangle =  \tilde{a}_0 |\tilde{\phi}_0 \rangle + \sum_{k=1}^\infty\left( 
		\tilde{a}_k |\tilde{\phi}_k \rangle + \hat{a}_k |\hat{\phi}_k \rangle \right) =
		\sum_{k=1}^\infty \hat{a}_k | s_k \rangle .
	\end{equation*}
\end{Rem}
\begin{Cor}
	The actions of the coin flip $R$ and the position swap $S$ on the vectors $|\sigma_k\rangle $  read
	\begin{equation*}
	R |\sigma_k\rangle = - |\sigma_k\rangle, \qquad S |\sigma_k\rangle = |\sigma_k\rangle, \qquad k \geq 0. 
	\end{equation*}
\end{Cor}

\section{The spectral measures of the random and quantum walks} \label{sec:Sz-CMV-spec}
In the previous section we identified the restriction of the quantum walk operator $U$ of Szegedy's quantization of the random walk on half-line to the subspace $\mathcal{H}_0$, generated by $U^k|\phi_0\rangle$, $k\in\ZZ$, with the CMV matrix $\mathcal{C}$. Moreover, due to equations \eqref{eq:Se} and \eqref{eq:Re}, the corresponding restrictions of the coin flip $R$ and the position swap $S$ have been identified with the matrices $\mathcal{M}$ and $\mathcal{L}$, correspondingly. These matrices are real and involutive $\mathcal{M}^2=\mathcal{L}^2 = \II$. The discriminant matrix of the quantum walk coincides with the Jacobi matrix ~\eqref{eq:J-CMG} of the initial random walk where the relation between the Jacobi and Verblunsky parameters are given by \eqref{eq:rk} and \eqref{eq:pq} written as
\begin{equation} \label{eq:sk}
s_k = \sqrt{p_k q_{k+1}} = \frac{1}{2}\sqrt{ (1-\alpha_{2k-1}) (1 - \alpha_{2k}^2) (1+\alpha_{2k+1})}. 
\end{equation}

Such equations are known~\cite{Simon-OPUC} as the Geronimus relations~\cite{Geronimus}. They go back to Szeg\H{o}~\cite{Szego} who gave a correspondence which, starting from polynomials orthogonal on the measure supported in the segment $[-1,1]$, allows to find certain natural measure on the unit circle together with the corresponding polynomials. 	
It turns out that the relation is analogous to that appearing~\cite{Szegedy} in Szegedy's quantization of Markov chains given in Proposition~\ref{prop:spec-Sz}.

\subsection{Szeg\H{o} projection and Geronimus relations}
Let us present first, following~\cite{KillipNenciu,Simon-OPUC}, geometric meaning of the Szeg\H{o} projection~\cite{Szego}.  For OPUC with real Verblunsky coefficients, or equivalently the measure $\mu$ being symmetric with respect to complex conjugation one can define the measure $\nu$ on the segment $[-1,1]$ by
\begin{equation*}
\int_{-1}^1 g(x) \,d\nu(x) = 
\int_{\partial\DD} g(\cos\theta)\, d\mu(\theta).
\end{equation*}
The relation between spectral measures 
\begin{equation} \label{eq:mu-nu}
d\mu(\theta) = w(\theta) \frac{d\theta}{2\pi} + d\mu_s , \qquad \text{and} \qquad d\nu(x) = u(x) dx + d\nu_s,
\end{equation}
has the form
\begin{equation} \label{eq:w-u}
u(x) = \frac{w(\arccos x)}{\pi \sqrt{1-x^2}}, \qquad
w(\theta) = \pi |\sin \theta | u (\cos\theta).
\end{equation}
The precise connection formulas between the polynomials orthogonal with respect to both measures~\cite{Szego} and relations between their Verblunsky and Jacobi coefficients are known~\cite{Geronimus}. 
\begin{Prop} \label{prop:polynomials}
	The polynomials orthonormal $p_k(x)$ with respect to the measure $d\nu(x)$ are expressed by the polynomials $\varphi_{k}(z)$ orthonormal with respect to the measure $d \mu(\theta)$ as follows
	\begin{equation} \label{eq:p-phi-Szego}
	p_k(x) = \frac{1}{\sqrt{2(1-\alpha_{2k-1})}} \left( z^{-k}\varphi_{2k}(z) + z^k \varphi_{2k}(z^{-1}) \right), \qquad x = \frac{1}{2}\left(z + z^{-1} \right).
	\end{equation}
\end{Prop}
We will give its new simple proof based on the structure of the CMV basis. For real Verblunsky coefficients we have $\mathcal{M}^2 = \mathcal{L}^2 = \mathbb{I}$. Denote by $\mathcal{H}^\mathcal{M}_\pm$ the eigenspace of $\mathcal{M}$ corresponding to eigenvalue $\pm 1$ and eigenvectors 
\begin{equation} \label{eq:c-pm1}
c^\pm_k = \sqrt{\frac{1\pm\alpha_{2k-1}}{2}} \, e_{2k-1} \pm \sqrt{\frac{1 \mp \alpha_{2k-1}}{2}} \, e_{2k}, \qquad k>0;
\end{equation}
compare with equation~\eqref{eq:phi-e}.
\begin{Rem}
	Notice that the above formula applies also to $c^+_0 = e_0$. 
\end{Rem}
The transition amplitudes due to application of $\mathcal{L}$ on the eigenvectors 
\begin{equation*}
\mathcal{L} c_k^+ = a^{++}_k c^+_{k-1} + a^{-+}_k c^-_{k-1} + b^{++}_k c^+_{k} + b^{-+}_k c^-_{k} + c^{++}_k c^+_{k+1} + c^{-+}_k c^-_{k+1},
\end{equation*}
read as follows
\begin{align*}
a^{++}_k = \langle c_{k-1}^+ | \mathcal{L} c_k ^+ \rangle & = \; \; \frac{1}{2} \sqrt{ (1-\alpha_{2k-3})(1-\alpha_{2k-2}^2)(1+\alpha_{2k-1})},\\
a^{-+}_k = \langle c_{k-1}^- | \mathcal{L} c_k ^+ \rangle & = - \frac{1}{2} \sqrt{ (1+\alpha_{2k-3})(1-\alpha_{2k-2}^2)(1+\alpha_{2k-1})} ,\\
b^{++}_k = \; \; \langle c_{k}^+ | \mathcal{L} c_k ^+ \rangle \; & = \; \; \frac{1}{2} \Big( \alpha_{2k} (1-\alpha_{2k-1}) - \alpha_{2k-2} (1+\alpha_{2k-1})  \Big) , 
\\
b^{-+}_k = \;\; \langle c_{k}^- | \mathcal{L} c_k ^+ \rangle \; & = - \frac{1}{2} \left(\alpha_{2k-2}  \sqrt{1-\alpha_{2k-1}^2} + \alpha_{2k-2} \sqrt{1-\alpha_{2k-1}^2}  \right) , \\
c^{++}_k = \langle c_{k+1}^+ | \mathcal{L} c_k ^+ \rangle & = \; \;\frac{1}{2} \sqrt{ (1-\alpha_{2k-1})(1-\alpha_{2k}^2)(1+\alpha_{2k+1})} = a^{++}_{k+1} , \\
c^{-+}_k = \langle c_{k+1}^- | \mathcal{L} c_k ^+ \rangle & = \; \; \frac{1}{2} \sqrt{ (1-\alpha_{2k-1})(1-\alpha_{2k}^2)(1-\alpha_{2k+1})} .	
\end{align*}
By direct calculation using the splitting $\mathcal{C} = \mathcal{L} \mathcal{M}$, reality consequences $\mathcal{L}^t = \mathcal{L}$, $\mathcal{M}^t = \mathcal{M}$, and the fact that $\mathcal{M}c^\pm_k = \pm c^\pm_k$ one obtains 
\begin{equation} \label{eq:JMLC}
\frac{1}{2} \left(\mathcal{C} + \mathcal{C}^t \right)c_k^+ = a^{++}_k c^+_{k-1} +  b^{++}_k c^+_{k} + a^{++}_{k+1} c^+_{k+1},
\end{equation}
The restriction of the above self-adjoint operator~\cite{KillipNenciu,Simon-OPUC} 
to $\mathcal{H}^\mathcal{M}_+$ has still $e_0$ as the cyclic vector, and is of the form
\begin{equation*} 
\mathcal{J_+} = \left( \begin{array}{ccccc}
r_0 & s_0 & 0 & 0 & \cdots \\
s_0 & r_1 & s_1 & 0 & \cdots \\
0 & s_1 & r_2 & s_2 & \cdots \\
0 & 0 & s_2 & r_3 & \cdots \\
\vdots & \vdots & \vdots & \vdots & \ddots
\end{array} \right),
\end{equation*}
where the coefficients $r_k= b_k^{++}$, $s_k = a_{k+1}^{++}$ are given by the Geronimus relations~\eqref{eq:rk} and \eqref{eq:sk}.

Recall that action of $\mathcal{J}_+$ is represented by multiplication by the variable $x$ in the restricted space, and by equation \eqref{eq:JMLC} it is represented by multiplication by $\frac{1}{2}(z+z^{-1})$ in the full space (both with $e_0$ as the cyclic vector).	
Due to equations~\eqref{eq:chi-phi} and \eqref{eq:c-pm1}, in the cyclic model the vectors $c_k^+$ of the orthonormal basis of $\mathcal{H}_+^\mathcal{M}$ are represented by the Laurent polynomials
	\begin{equation} 
	\sqrt{\frac{1+\alpha_{2k-1}}{2}} \chi_{2k-1}(z) + \sqrt{\frac{1-\alpha_{2k-1}}{2}} \chi_{2k}(z) =  \sqrt{\frac{1+\alpha_{2k-1}}{2}} z^{-k} \varphi^*_{2k}(z) + \sqrt{\frac{1-\alpha_{2k-1}}{2}} z^{-k} \varphi_{2k+1}(z).
	\end{equation}
	The above representation can be simplified by using the recurrence formulas~\eqref{eq:phi-rec} which imply
	\begin{equation*}
	\varphi_{2k+1}(z) = \frac{1}{\rho_{2k}} (z \varphi_{2k}(z) - \bar{\alpha}_{2k} \varphi^*_{2k}(z)),
	\end{equation*}
	and by reality of the Verblunsky coefficients which gives 
	\begin{equation*}
	\varphi^*_{2k}(z) = z^{2k} \varphi_{2k}(z^{-1}).
	\end{equation*}
Then the expressions for the polynomial representation of the basis vectors~\eqref{eq:c-pm1} take the form of the right hand side of equation \eqref{eq:p-phi-Szego}. Invariance of the Laurent polynomial with respect to $z\leftrightarrow z^{-1}$ implies that it is a polynomial in the variable $z+z^{-1}$. This also gives the new simple proof of Proposition~\ref{prop:polynomials}. 

\begin{Rem}
Similarly, the action of $\mathcal{L}$ on the second eigenvectors $(c_k^-)_{k\in\NN}$, which span $\mathcal{H}_-^\mathcal{M}$, by formula analogous to \eqref{eq:JMLC}, gives another Jacobi matrix for the cyclic vector $c_1^-$. The corresponding measure and polynomials were found by Szeg\H{o}~\cite{Szego}. By symmetry reason, one can apply the same reasoning for eigenspaces $\mathcal{H}_-^\mathcal{L}$ of the matrix $\mathcal{L}$. The corresponding measures were found in~\cite{KillipNenciu}, and the orthogonal polynomials were found in~\cite{BerriochaCachafeiroGarcia-Amor} from the recurrence relations and Jacobi matrices. We mention that the approach via CMV matrices, as presented in this section, gives the same result with less effort. Finally, without entering into details we remark that the above modifications of both the measure and polynomials can be incorporated into more general context of Darboux-type (Christoffel and Geronimus) transformations of CMV matrices~\cite{CMMV}. 
\end{Rem}
\begin{Rem}
	In literature one can find~\cite{DerevyaginVinetZhedanov} also another map between OPUC and OPRL based on the theory of CMV matrices.
\end{Rem}

\subsection{Example of Jacobi polynomials}
The celebrated Jacobi polynomials are orthogonal with respect to the measure on $[-1,1]$
\begin{equation*}
d\nu(x) = (1-x)^\alpha (1+x)^\beta dx , \qquad \alpha, \beta > -1,
\end{equation*}
and for certain values of the parameters reduce to the Gegenbauer $(\alpha = \beta)$, Legendre $(\alpha = \beta = 0)$, Chebyshev polynomials of the first $(\alpha = \beta = - \frac{1}{2})$ or the second $(\alpha = \beta = \frac{1}{2})$ kind. Also other well known families of orthogonal polynomials named after Laguerre and Hermite can be derived as certain limiting cases of the Jacobi polynomials.

When defined by the Rodrigues formula
\begin{equation*}
P^{(\alpha,\beta)}_n(x) = \frac{(-1)^n}{n! 2^n} (1-x)^{-\alpha}(1+x)^{-\beta} \left( (1-x)^{\alpha+n}(1+x)^{\beta + n}
	\right)^{(n)}, \qquad n\geq 0,
\end{equation*}
they are normalized by
\begin{gather*}
\int_{-1}^1 P^{(\alpha,\beta)}_m(x) P^{(\alpha,\beta)}_n(x) 
(1-x)^{\alpha}(1+x)^{\beta} dx = \delta_{mn} h^{(\alpha,\beta)}_n, \\
h^{(\alpha,\beta)}_n = \frac{2^{\alpha + \beta + 1} \Gamma(\alpha + n + 1) \Gamma(\beta + n + 1)}{n! \Gamma(\alpha + \beta + n + 1) \Gamma(\alpha + \beta + 2n + 1)}.
\end{gather*}
The leading coefficient is given by
\begin{equation*}
P^{(\alpha,\beta)}_n(x) = \frac{1}{2^n} \begin{pmatrix} \alpha + \beta + 2n \\ n \end{pmatrix} x^n + \dots ,
\end{equation*}
and their values at the ends of the segment read
\begin{equation*}
P^{(\alpha ,\beta)}_n(1) = \begin{pmatrix} \alpha + n \\ n \end{pmatrix}, \qquad P^{(\alpha ,\beta)}_n(-1) = (-1)^n \begin{pmatrix} \beta + n \\ n \end{pmatrix} .
\end{equation*}
The three-term recurrence relation is of the form
\begin{gather*}
\begin{split}
x P^{(\alpha ,\beta)}_n(x) = 
\frac{2(n+1)(n+\alpha + \beta + 1)}{(2n + \alpha + \beta + 1)(2n + \alpha + \beta + 2)} P^{(\alpha ,\beta)}_{n+1}(x) +
\qquad \qquad \\ +
\frac{\beta^2 - \alpha^2}{(2n + ,\alpha + \beta)(2n + \alpha + \beta + 2)} P^{(\alpha ,\beta)}_{n}(x) + 
\frac{2(n+\alpha)(n + \beta)}{(2n + \alpha + \beta)(2n + \alpha + \beta + 1)} P^{(\alpha ,\beta)}_{n-1}(x),
\end{split}\\
P^{(\alpha ,\beta)}_{-1}(x) \equiv 0, \qquad 
P^{(\alpha ,\beta)}_0(x) \equiv 1.
\end{gather*}
The corresponding polynomials defined by
\begin{equation*}
Q_n^{(\alpha,\beta)}(x) = \frac{P^{(\alpha ,\beta)}_n(x)}{P^{(\alpha ,\beta)}_n(1)},
\end{equation*}
satisfy three term recurrence
\begin{equation*}
x Q_n^{(\alpha,\beta)}(x) = p_n^{(\alpha,\beta)} Q_{n+1}^{(\alpha,\beta)}(x) + r_n^{(\alpha,\beta)} Q_{n}^{(\alpha,\beta)}(x) + q_n^{(\alpha,\beta)} Q_{n-1}^{(\alpha,\beta)}(x) ,
\end{equation*}
with the coefficients
\begin{align} \label{eq:pk-J}
p_n^{(\alpha,\beta)} & = \frac{2(n+\alpha + 1)(n+\alpha + \beta + 1)}{(2n + \alpha + \beta + 1)(2n + \alpha + \beta + 2)} ,\\
r_n^{(\alpha,\beta)} & = \frac{\beta^2 - \alpha^2}{(2n + \alpha + \beta)(2n + \alpha + \beta + 2)} , \\ \label{eq:qk-J}
q_n^{(\alpha,\beta)} & = \frac{2n(n+ \beta)}{(2n + \alpha + \beta)(2n + \alpha + \beta + 1)},
\end{align}
where also 
\begin{equation*}
p_n^{(\alpha,\beta)} + r_n^{(\alpha,\beta)} + q_n^{(\alpha,\beta)} = 1, \qquad p_0^{(\alpha,\beta)} > 0, \quad q_0^{(\alpha,\beta)} = 0, \qquad p_n^{(\alpha,\beta)} > 0, \qquad q_n > 0 \quad \text{for} \quad n>0.
\end{equation*}
In order to be random walk polynomials they have to satisfy the recurrence with $r_n^{(\alpha,\beta)}\geq 0$, which needs $\alpha = \beta$ or $\beta \geq |\alpha|$. 

The corresponding quantum walks are governed by circular analogs of the Jacobi polynomials~\cite{Szego}, obtained by the Szeg\H{o} transform, are given thus by the weight
\begin{equation*}
w(\theta) = (1-\cos \theta)^{\alpha + 1/2} (1+\cos\theta)^{\beta + 1/2}, \qquad \alpha, \beta > -1, \qquad \theta\in[0,2\pi].
\end{equation*}
The Verblunsky coefficients for the measure have been found in~\cite{Golinskii-J,Badkov} and read as follows
\begin{equation} \label{eq:circ-J-V}
\alpha_{n} = - \frac{\alpha + \frac{1}{2} + (-1)^{n+1}(\beta + \frac{1}{2})}{n + \alpha + \beta + 2}.
\end{equation}	
Indeed, inserting the coefficients \eqref{eq:circ-J-V} into formulas~\eqref{eq:qk}-\eqref{eq:pk} we recover the above probabilities \eqref{eq:pk-J}-\eqref{eq:qk-J}.

\subsection{Direct (or naive) derivation of quantum walks from CMV matrices}
The $\mathcal{L}\mathcal{M}$ factorization of CMV matrices suggests yet another derivation of quantum walks on half-line. At first glance it looks promising, but after presenting it we would like to give arguments towards its rejection. 

The key idea is to interpret one of the matrices as providing the coin flip operator with the second acting as the position flip operator. 
Consider coins having two states over vertex $k > 0$ with the basis $\{ e_{2k-1}, e_{2k} \}$, and one-dimensional state space over $k=0$ spanned by $e_0$. The matrix $\mathcal{M}$ gives then coin flips, in particular $\Theta_{2k-1}$ flips the coin over vertex $k>0$. The matrix $\mathcal{L}$ flips positions, in particular $\Theta_{2k}$ provides flip between the coin state $e_{2k}$ over vertex $k$ and the state $e_{2k+1}$ over vertex $k+1$, see Figure~\ref{fig:CMV-walk-2}.
\begin{figure}[h!]
	\begin{center}
		\includegraphics[width=12cm]{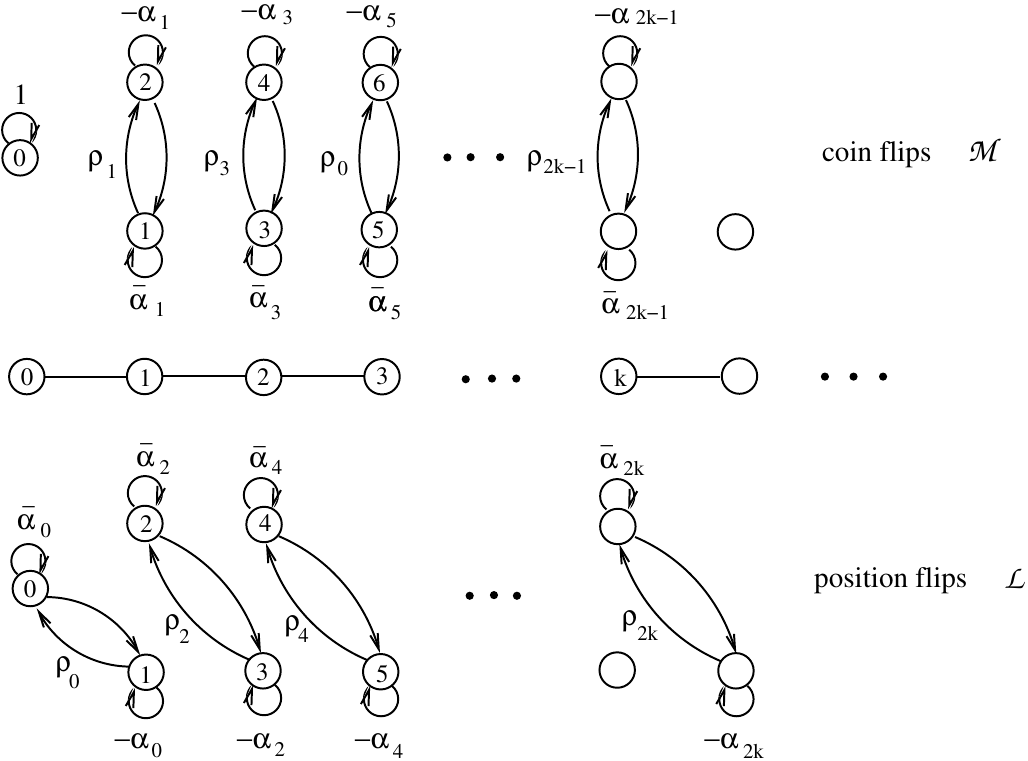}
	\end{center}
	\caption{Quantum walk on the half-line splitted into coin flips described by matrix $\mathcal{M}$ and position flips given by $\mathcal{L}$}
	\label{fig:CMV-walk-2}
\end{figure}
In order to measure the transition probabilities from vertex $k$ of the half-line graph we need to fix initial state of the coin there. When a choice of a coin in each vertex has been made the formulas below provide a way how to associate a discrete random walk on a half line with a given quantum walk with $\mathcal{C}=\mathcal{L}\mathcal{M}$ serving as unitary evolution operator. 
\begin{Lem} \label{lem:prob-2}
	When the particle is in vertex $k$ of the half-line graph with the coin in the state 
	\begin{equation*}
	c_k =	\gamma_{2k-1}e_{2k-1} + \gamma_{2k} e_{2k}, \qquad |\gamma_{2k-1}|^2 + |\gamma_{2k}|^2 = 1, \qquad k=1,2,3,\dots ,
	\end{equation*}
	then after application of the position flip $\mathcal{L}$ the particle
	\begin{enumerate}
		\item moves to vertex $k-1$ with probability
		\begin{equation*}
		q_k = |\gamma_{2k-1}|^2 (1-|\alpha_{2k-2}|^2),
		\end{equation*}
		\item moves to vertex $k+1$ with probability
		\begin{equation*}
		p_k = |\gamma_{2k}|^2 (1 - |\alpha_{2k}|^2),
		\end{equation*}
		\item stays at vertex $k$ with probability
		\begin{equation*}
		r_k = |\gamma_{2k-1}|^2 |\alpha_{2k-2}|^2 + |\gamma_{2k}|^2 |\alpha_{2k}|^2.
		\end{equation*}
	\end{enumerate}
\end{Lem}
\begin{proof}
	After the position flip the state of the particle-coin reads
	\begin{equation} \label{eq:p-flip-c}
	\mathcal{L}c_k = \gamma_{2k-1}\rho_{2k-2} e_{2k-2} - \gamma_{2k-1} \alpha_{2k-2} e_{2k-1} + \gamma_{2k} \bar\alpha_{2k} e_{2k} + \gamma_{2k} \rho_{2k} e_{2k+1} ,
	\end{equation}
	which implies the statement.
\end{proof}
\begin{Rem}
	Notice that the above transition probabilities do not depend either on the phases of the Verblunsky coefficients $\alpha_{2k-2}$ and $\alpha_{2k}$ or on the phases of the initial state of the coin. We considered the action of the position flip $\mathcal{L}$ only, because the action of the coin flip can be compensated by different initial state.
\end{Rem}

In order to compare quantum walks with discrete random walks one needs to establish a correspondence between them. It is natural to demand the same one-step transition probabilities. In Szegedy's quantization this requirement is satisfied due to Proposition~\ref{prop:Szegedy-Markov}. Because dependence of the result on the state of the coin it is important to include this additional degree of freedom into the definition of the correspondence.

Let us assume that all the Verblunsky coefficients are real and let us take as the initial coin state the eigenvector of the local coin flip operator $\mathcal{M}$ with eigenvalue~$1$, like in the spectral quantization.  Then in Lemma~\ref{lem:prob-2} we take
\begin{equation*}
c_k = \sqrt{\frac{1 + \alpha_{2k-1}}{2}} \, e_{2k-1} + \sqrt{\frac{1 - \alpha_{2k-1}}{2}} \, e_{2k},
\end{equation*}
which results in the transition probabilities
	\begin{align} \label{eq:qk-direct}
q_k & = \frac{1}{2}(1+\alpha_{2k-1})(1-\alpha_{2k-2}^2)\\
p_k & = \frac{1}{2}(1-\alpha_{2k-1})(1-\alpha_{2k}^2)\\	
r_k & = \frac{1}{2}\left( \alpha_{2k-2}^2(1+\alpha_{2k-1}) + \alpha_{2k}^2 (1-\alpha_{2k-1}) \right),
\label{eq:rk-direct}
\end{align}
with the convention $\alpha_{-1} = -1$.
\begin{Rem}
	The above relations between probabilities and Verblunsky coefficients differ from those obtained in the spectral quantization \eqref{eq:qk}-\eqref{eq:pk}. They agree only for vanishing even Verblunsky coefficients $\alpha_{2k}=0$, see Example~\ref{ex:r=0}.
\end{Rem}
\begin{Rem}
Equations \eqref{eq:qk-direct}-\eqref{eq:rk-direct} can be reversed, and with the initial condition $\alpha_0 = \sqrt{r_0}$ other Verblunsky coefficients can be computed recursively. However the equations impose restrictions on the transition probabilities. For example one obtains
\begin{equation*}
\frac{q_1}{p_0} = \frac{1+\alpha_1}{2},
\end{equation*}
which implies that in order to have $\alpha_1 < 1$ one needs $q_1<p_0$.
\end{Rem}

\section{Quantization of discrete random walks with constant transition probabilities} \label{sec:constant-prob}
In this section we present in detail spectral quantization of random walks with constant transition probabilities. The corresponding orthogonal polynomials on the unit circle have real Verblunsky coefficients of period two. In order to present their properties we study more general problem of orthogonal polynomials with complex coefficients of such period. These results motivate our findings concerning polynomials with coefficients of arbitrary period, which we give in the Appendix.

\subsection{Orthogonal polynomial representation of discrete random walks with constant transition probabilities}
	Let us consider the random walk with probabilities
	\begin{equation} \label{eq:pqr-const}
	p_0 , \qquad r_0 = 1-p_0 , \qquad p_k  = p, \qquad q_k = q, \qquad r_k = 1 - p - q, \qquad k\geq 1.
	\end{equation}
Such a choice is motivated by the example given in the classical paper \cite{KarlinMcGregor} and by the following result, which can be checked by direct calculation using Proposition~\ref{prop:pqr-a}.
\begin{Prop} \label{prop:pqr-2p}
The transition coefficients of the random walk which leads by spectral quantization to two-periodic Verblunsky coefficients $\alpha_{2k} = a$, $\alpha_{2k+1} = b$, $k=0,1,2,\dots$ are given by \eqref{eq:pqr-const} with
\begin{gather*}
p_0 = 1-a, \qquad r_0 = a, \\
p = \frac{1}{2}(1-a)(1-b), \qquad r = -ab, \qquad q = \frac{1}{2} (1+a)(1+b).
\end{gather*} 
\end{Prop}	
The Cauchy--Stieltjes transform of the measure for coefficients given by \eqref{eq:pqr-const} was calculated	in~\cite{GdI}  and reads 
\begin{equation*}
S(z) = \frac{-2p(z-1+p_0) - p_0(1-p-q-z) + p_0 \sqrt{(1-p-q-z)^2 - 4pq}}{2(p_0 - p) (1-z) (z-\xi)},
\end{equation*}
where 
\begin{equation*}
\xi = 1 - p_0 - \frac{p_0 q}{p_0 - p}.
\end{equation*}
Then the absolutely continuous part of the measure is given by
\begin{equation*}
u(x) = 
\frac{p_0 \sqrt{(x-\sigma_-)(\sigma_+-x)}}{2\pi q(p_0 - p)(1-x)(x-\gamma)} , \qquad x\in [\sigma_-, \sigma_+] \subset [-1,1], \qquad \sigma_\pm = 1 - (\sqrt{p} \mp \sqrt{q})^2,
\end{equation*}
and the discrete part is concentrated at $x=1$, and $x=\gamma$
\begin{equation*}
d\nu_s =   \nu(\{1\}) \delta(x-1) +  \nu(\{\xi\}) \delta(x-\xi),
\end{equation*}
with non-zero mass
\begin{equation*}
\nu(\{1\}) = \frac{q-p}{q-p + p_0}
\end{equation*} 
at $x=1$ existing for $q-p>0$ only, and at $x=\xi$ when $p_0 \neq p$
\begin{equation*}
\nu(\{ \xi \}) = \frac{(p_0 - p)^2 - pq}{(p_0 - p)(p_0 - p + q)}
\end{equation*}	
for $(p_0 - p)^2 > pq$ only.
\begin{Cor} \label{cor:u-ab}
	In the special case described in Proposition~\ref{prop:pqr-2p} the corresponding formulas read
	\begin{gather}
	u(x) = \frac{\sqrt{(x-\sigma_-)(\sigma_+ - x)}}{\pi (1+b) (1-x)(x+1)}, \qquad \text{where} \quad \sigma_\pm  = -ab \pm \rho_a \rho_b , \\
	\nu(\{ 1 \}) = \frac{a+b}{1+b}, \quad \text{when} \; a+b>0,  \qquad \xi = -1, \quad \nu(\{-1\}) = \frac{b-a}{1+b}  \quad \text{when} \; b>a.
	\end{gather}
\end{Cor}
The orthogonal polynomials corresponding to the random walk \eqref{eq:pqr-const} satisfy recurrence
	\begin{align*}
	xQ_k(x) & = q Q_{k-1}(x) + (1-p-q) Q_k(x) + p Q_{k+1}(x), \qquad k\geq 1,\\
	Q_0(x) & \equiv 1, \qquad Q_1(x) = \frac{x- 1 + p_0}{p_0}.
	\end{align*}
Upon introducing 
	\begin{equation*}
	y = \frac{x-r}{2\sqrt{pq}}, \qquad \tilde{Q}_k(y) = \left( \frac{p}{q} \right)^{k/2} Q_k (x),
	\end{equation*}
we obtain the recurrence 
	\begin{gather} \label{eq:Ch-rec-Q}
	2 y \tilde{Q}_k(y) = \tilde{Q}_{k+1}(y) + \tilde{Q}_{k-1}(y), \\
	\tilde{Q}_0(y) \equiv 1, \qquad \tilde{Q}_1(y) = \frac{2p}{p_0}y + \left( \frac{q}{p} \right)^{1/2} \frac{p_0 - p - q}{p_0}.
	\end{gather}
It can be solved in terms of the Chebyshev polynomials of the second kind $U_k(y)$, defined by the same recurrence \eqref{eq:Ch-rec-Q} but with initial conditions
	\begin{equation*}
	U_0(y) \equiv 1, \qquad U_1(y) = 2y ,
	\end{equation*}
and leads to the final result
	\begin{equation} \label{eq:OPRL-1p}
	Q_k(x) = \left(\frac{q}{p} \right)^{k/2} \left[ U_k\left(\frac{x-r}{2\sqrt{pq}}\right) +  \left(\left( \frac{p}{q} \right)^{1/2}\frac{x-1 + p_0}{p_0} - \frac{x-r}{\sqrt{pq}}\right) U_{k-1}\left(\frac{x-r}{2\sqrt{pq}}\right)\right], \qquad k\geq 0.
	\end{equation}
\begin{Rem}
	We use the fact that the recurrence for Chebyshev polynomials of the second kind can be equivalently initiated by 
	\begin{equation*}
U_{-1}(y)\equiv 0, \qquad U_0(y) \equiv 1.
	\end{equation*}
Another standard definition of the polynomials, which we will use in the next section, is given by
\begin{equation}
\label{eq:Cheb-sec-sin}
U_k(y) = \frac{\sin ((k+1)\omega)}{\sin\omega}, \qquad y = \cos \omega .
\end{equation}
\end{Rem}
	\begin{Rem}
	See Section~\ref{sec:per-OPRL} for generalization of the above representation to the case of OPRL with coefficients of arbitrary period.
	\end{Rem}
\begin{Cor} \label{cor:P}
	The polynomials 
\begin{equation*}
P_k(x) = \left( \frac{p}{q} \right)^{k/2} Q_k (x)
\end{equation*}
satisfy recurrence
\begin{gather*}
xP_k(x) = \sqrt{pq}P_{k+1}(x) + r P_k(x) + \sqrt{pq}P_{k-1}(x), \qquad k\geq 1 ,\\
P_0(x) \equiv 1, \qquad P_1(x) = \left( \frac{p}{q} \right)^{1/2} \frac{x-1+p_0}{p_0},
\end{gather*}
governed by the symmetric Jacobi matrix~\eqref{eq:J-CMG} constructed from the random walk with coefficients~\eqref{eq:pqr-const}. 	In the special case described in Proposition~\ref{prop:pqr-2p} the corresponding recurrence reads
\begin{gather*}
xP_k(x) = \frac{1}{2}\rho_a \rho_b P_{k+1}(x) -ab P_k(x) + \frac{1}{2}\rho_a \rho_b P_{k-1}(x), \qquad k\geq 1,\\
P_0(x) \equiv 1, \qquad P_1(x) = \frac{1-b}{\rho_a \rho_b}(x-a),
\end{gather*}
and gives polynomials
\begin{equation} \label{eq:Pk}
P_k(x) =  U_k\left(\frac{x+ab}{\rho_a \rho_b}\right) - \frac{(1+b)(x+a)}{\rho_a \rho_b} U_{k-1}\left(\frac{x+ab}{\rho_a \rho_b}\right) , \qquad k\geq 0.
\end{equation}
\end{Cor}

\subsection{Orthogonal polynomials with two-periodic Verblunsky coefficients} \label{sec:OPUC-2p}
We will study in detail orthogonal polynomials on the unit circle with two-periodic Verblunsky coefficients. Although we need such polynomials with real coefficients only, we consider first general case of complex coefficients. 
\subsubsection{The spectral measure}
Because the spectral measure is well described in literature~\cite{GeronimoVanAssche,Simon-OPUC,Khrushchev} in the general multi-periodic case we will start with brief presentation of the corresponding results in the two-periodic case, i.e. with Verblunsky coefficients of the form
\begin{equation*}
\alpha_{2k} = a, \qquad \alpha_{2k+1} = b, \qquad k= 0,1,2,\dots \, , \qquad a,b \in \DD.
\end{equation*}
The Schur function $f(z)$ is given then by two-periodic continued fraction and satisfies quadratic equation
\begin{equation*}
(\bar{a}z + \bar{b})[f(z)]^2 + \left( \frac{1}{z} + \bar{a} b - a \bar{b} - z \right) f(z) - \frac{a}{z} - b = 0.
\end{equation*}
The corresponding Carath\'{e}odory function $F(z)$ satisfies analogous equation with solution
\begin{equation*}
F(z) = \frac{-B(z) + \sqrt{B(z)^2 - 4A(z)C(z)}}{2A(z)}
\end{equation*} 
where
\begin{gather*}
A(z)  = - z^2 (b + 1) - z (a + a\bar{b} - \bar{a} - \bar{a} b) + 1 + \bar{b},\quad
B(z)  = -2 (b z^2 + (a + \bar{a})z +\bar{b}), \\
C(z)  =  - z^2 (b - 1) - z (a - a\bar{b} - \bar{a} + \bar{a} b) - 1 + \bar{b},
\end{gather*}
and with the sign fixed by condition $F(0) = 1$. 
This allows to find the absolutely continuous part of the measure, which equals
\begin{equation} \label{eq:2p-w}
w(\theta) = \frac{\sqrt{\rho_a^2 \rho_b^2 - [\mathrm{Re}(e^{i\theta} + a\bar{b})]^2}}{|\mathrm{Im}(e^{i\theta}-\bar{a})(1+b)|}  \qquad  \text{when} \quad 
\theta \in (\theta_+, \theta_-) \cup (2\pi - \theta_{-}, 2\pi - \theta_+)
\end{equation}
where $\theta_{\pm}\in [0,\pi]$, 
$\cos\theta_\pm = \pm\rho_a \rho_b - \mathrm{Re}(a \bar{b})$, and vanishes otherwise; here $\rho_a = \sqrt{1-|a|^2}$, $ \rho_b = \sqrt{1 - |b|^2}$.

The eventual points $z_\pm$ of the discrete spectrum are roots of the equation $A(z) = 0$ and read
\begin{equation*}
z_\pm = e^{i\gamma_\pm} = \frac{-(a + a\bar{b} - \bar{a} - \bar{a} b) \pm \sqrt{(a + a\bar{b} - \bar{a} - \bar{a} b)^2 + 4 |1+b|^2}}{2(b+1)}
\end{equation*}
The point $z_\pm$ is in the discrete spectrum if and only if $\mathrm{Re} (a + b z_{\pm}) \gtrless 0$, and the corresponding weight reads
\begin{equation*}
\mu(\{ z_\pm\}) =\left| \frac{2 \mathrm{Re} (a + b e^{i\gamma_{\pm}})}{ (b + 1) \sin [(\gamma_+ - \gamma_-)/2]}\right|.
\end{equation*}

\begin{Rem}
	In \cite{PeherstorferSteinbauer,Simon-OPUC} one can find general expressions for the spectral measure in the periodic case which our result is a special case. 
\end{Rem}

\begin{Cor}
In the special case of real Verblunsky coefficients $a,b\in \RR$ the identification $
\sigma_\pm = \cos\theta_\pm$
and application of equation~\eqref{eq:mu-nu} the absolutely continuous part of the measure \eqref{eq:2p-w} on the circle give the absolutely continuous part of the measure on the segment, as described in Corollary~\ref{cor:u-ab}. As for the discrete spectrum we have $z_\pm = \pm 1$ and $\mu(\{\pm 1\}) = 2 \nu(\{\pm 1\})$ as it should be expected.
\end{Cor}

\subsubsection{The polynomials}
In the literature the case of constant complex Verblunsky coefficients (one-periodic) leads to the so called Geronimus polynomials~\cite{Geronimus}. It is known that such polynomials can be described in terms of the Chebyshev polynomials~\cite{Golinskii} of the second kind. We will generalize the construction to the two-periodic case.

In order to use the recurrence \eqref{eq:phi-rec} to find the polynomials we diagonalize the transfer matrix
\begin{equation*}
T_2(z) = A(b)A(a) = \frac{1}{\rho_a \rho_b}\left( \begin{array}{cc}
z^2 + a \bar{b} z & -\bar{a} z - \bar{b} \\
-b z^2 - a z & \bar{a} b z +1
\end{array} \right).
\end{equation*}
Its eigenvalues read
\begin{equation*}
\lambda_\pm(z) = \frac{z}{2}\left[ \Delta_2(z) \pm \sqrt{\Delta_2^2(z) - 4 } \right],
\end{equation*}
where
\begin{equation*}
\Delta_2 (z) = \frac{1}{z} \mathrm{Tr}[T_2(z)]= \frac{1}{\rho_a \rho_b}\left( z + a \bar{b} + \bar{a} b + \frac{1}{z} \right)
\end{equation*}
is called the discriminant. The polynomials are given then by
\begin{equation*}
\varphi_{2k}(z) = A_+(z) \lambda_+(z)^k + A_-(z) \lambda_-(z)^k, \qquad \varphi_{2k+1}(z) = B_+(z) \lambda_+(z)^k + B_-(z) \lambda_-(z)^k,
\end{equation*}
where $A_\pm(z)$, $B_\pm(z)$ do not depend on $k$, and can be expressed in terms of $\varphi_{j}(z)$, $j=0,1,2,3$.

It is convenient to define the function $\omega(z)$ by $\cos \omega(z) = \frac{\Delta_2(z)}{2}
$,
which gives $\lambda_\pm(z) = z e^{\pm i \omega(z)}$ and allows to write
\begin{equation*}
\frac{\lambda_+(z)^k - \lambda_-(z)^k}{\lambda_+(z) - \lambda_-(z)} =
z^{k-1}\frac{e^{i k\omega(z)} - e^{-i k\omega(z)}}{e^{i\omega(z)} - e^{-i\omega(z)}} = z^{k-1}U_{k-1}(\cos\omega(z)), 
\end{equation*}
where we used one of standard definitions~\eqref{eq:Cheb-sec-sin} of the Chebyshev polynomials of the second kind.
After simple calculation which involves also the Chebyshev recurrence, we obtain the expressions 
\begin{align} \label{eq:2p-U}
\varphi_{2k}(z) & = z^k U_k(\cos\omega(z)) - z^{k-1} U_{k-1}(\cos\omega(z)) \frac{1+\bar{b} + z\bar{a}(1+b)}{\rho_a \rho_b},\\
\varphi_{2k+1}(z)  &= z^k U_k(\cos\omega(z)) \frac{ z - \bar{a}}{\rho_a}- z^{k}U_{k-1}(\cos\omega(z)) \frac{1+\bar{b}}{\rho_b} .
\end{align}

\begin{Cor}
	In the special case of real Verblunsky coefficients
	\begin{equation*}
	\frac{\Delta_2(e^{i\theta})}{2} = \frac{\cos\theta + ab}{\rho_a \rho_b},
	\end{equation*}
and the orthonormal polynomials on the segment calculated by the Szeg\H{o} formula \eqref{eq:p-phi-Szego} from the polynomials given by~\eqref{eq:2p-U} agree, up to a normalization factor $\sqrt{2/(1-b)}$ for $k>0$, with the polynomials $P_k(x)$ as described by \eqref{eq:Pk} in Corollary~\ref{cor:P}.
\end{Cor}

\subsubsection{Geometric construction of the spectrum}
\begin{figure}[h!]
	\begin{center}
		\includegraphics[width=12cm]{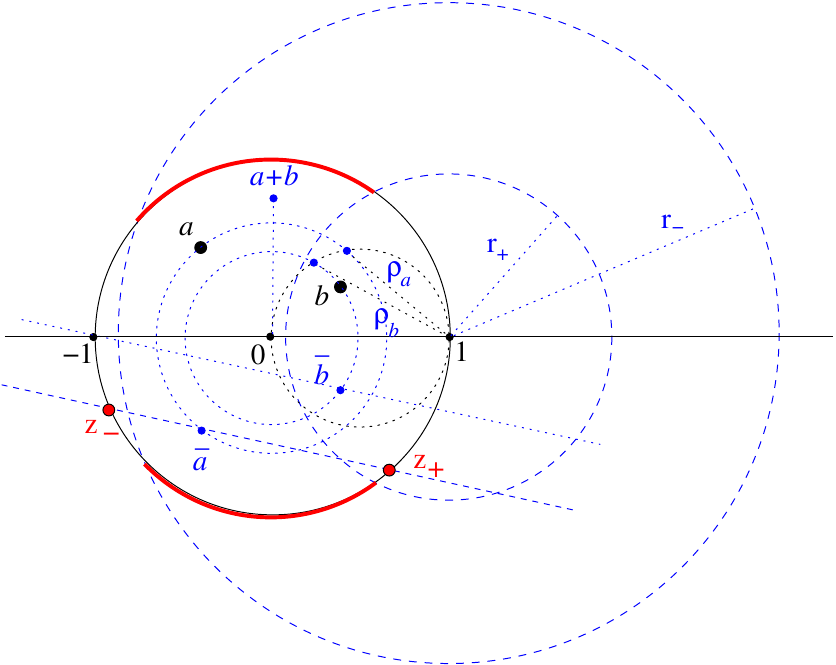}
	\end{center}
	\caption{Geometric construction of the spectrum of OPUC with two-periodic Verblunsky coefficients $a$ and $b$}
	\label{fig:spectrum}
\end{figure}
In the case $p=1$, i.e. for constant Verblunsky coefficients,
there exists~\cite{Geronimus} a simple geometric construction of  the edges of the continuous spectrum on the unit circle and of the discrete spectrum. Let us present analogous construction for $p=2$. 

\begin{Prop} The bands of the continuous spectrum for two-periodic Verblunsky coefficients are symmetrical with respect to the real axis. The edges $e^{\pm i \theta_+}$  (or $e^{\pm i \theta_-}$) are intersection points of the unit circle $\partial\DD$ with the circle centered in $1$ and with radius $r_+ = \sqrt{|a + b|^2 + (\rho_a - \rho_b)^2}$ (with radius $r_- = \sqrt{|a + b|^2 + (\rho_a + \rho_b)^2}$ respectively). The points $z_\pm$ of the discrete spectrum are intersections of the unit circle $\partial\DD$ and the line through the points $\bar{a}$ and $1 + \bar{a} + \bar{b}$.
\end{Prop}
\begin{proof}
	To show the first part rewrite equation defining the edges of the continuous spectrum as
	\begin{equation*}
	|1-z|^2 = |a + b|^2 + (\rho_a \pm \rho_b)^2.
	\end{equation*}
	The second part is a consequence of the equation satisfied by the points of the discrete spectrum
	\begin{equation*}
	\mathrm{Im}[(1+b)z + a(1+\bar{a})] = 0,	\end{equation*}
	which defines the line.
\end{proof}

\begin{Cor}
	The bands of the continuous spectrum and the discrete spectrum points can be constructed from $a$ and $b$ using ruler and compass. For example, by the Pythagoras theorem, $\rho_a$ is the length of the segment between $1$ and the intersection point of the circle $C(0,|a|)$ with the circle $C(1/2,1/2)$. Similarly, $|a + b|$, $|\rho_a \mp \rho_b |$ and $r_\pm$ are sides of a rectangular triangle. The procedure  reads as follows:
		\begin{enumerate}
		\item from $|a|$ and $|b|$ find $\rho_a$ and $\rho_b$,
		\item from $|a + b|$, $\rho_a$,  $\rho_b$ find $r_+$ and $r_-$ and construct edges of the continuous spectrum,
		\item draw the line through $\bar{a}$ in direction of $1+\bar{b}$ to find eventual points $z_+$ and $z_-$ of the discrete spectrum.
	\end{enumerate}
\end{Cor}
\begin{Rem}
	The edges $e^{\pm i \theta_+}$ (or $e^{\pm i \theta_-}$ ) of the continuous spectrum can be constructed equivalently as intersection points of $\partial\DD$ and the line through $-1$ and tangent to the circle with the center at $0$ and halved radius $r_+/2$ ( respectively $r_-/2$).
\end{Rem}
\begin{Rem}
	When $a=b$ then we get the Geronimus construction, see Figure~\ref{fig:spect-G}. The edges of the continuous spectrum are intersection of $\partial \DD$ with the circle centered at $1$ and with the radius $r_+ = 2|a|$; equivalently they are second intersection points of as intersection points of $\partial\DD$ and the lines through $-1$ and tangent to the circle with the center at $0$ and radius $|a|$.  The two branches of continuous spectrum close then (we have $r_- = 2$) at $z_- = -1$. The eventual point of the discrete spectrum
	\begin{equation*}
	z_+ = \frac{1+\bar{a}}{1+a}
	\end{equation*} 
	is the second intersection point of $\partial\DD$ and the line through $-1$ and $\bar{a}$.
\end{Rem}
\begin{figure}[h!]
	\begin{center}
		\includegraphics[width=8cm]{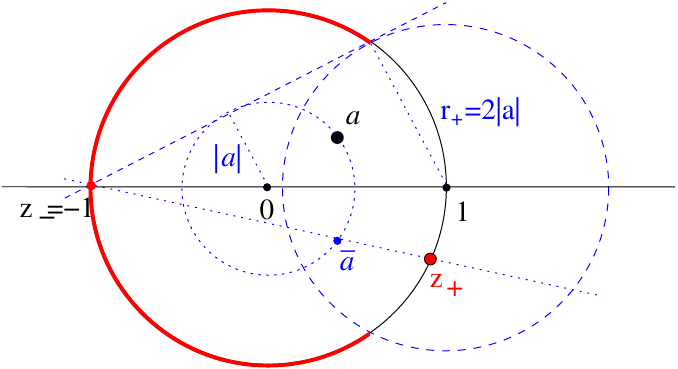}
	\end{center}
	\caption{Geometric construction of the spectrum of OPUC with constant Verblunsky coefficients $a$}
	\label{fig:spect-G}
\end{figure}
\section{Conclusions}
We gave a quantization scheme for discrete-time random walks on the half-line. It results in the well-known Szeg\H{o} map between polynomials orthogonal on the $[-1,1]$ segment and the polynomials orthogonal on the unit circle. The key-point of the scheme is the relation between the spectral measures of both polynomial systems. 

We didn't touch in this work the relation of the theory of orthogonal polynomials with integrable systems. Orthogonal polynomials on the unit circle are closely related to the Ablowitz--Ladik  hierarchy of integrable equations~\cite{AblowitzLadik}, see also another interpretation~\cite{OLPUC-Manas}, in the same way as the orthogonal polynomials on the real line are related with the Toda chain equations~\cite{Moser,Toda-TL}. It would be therefore interesting to apply the techniques of the theory of integrable systems to study quantum walks. In particular, the geometric meaning of the Ablowitz--Ladik hierarchy in terms of discrete curves~\cite{DS-AL} could provide an insight into possible realization of quantum information processing with the help of discrete curves such as RNA or DNA chains. Our spectral description of quantum walks in terms of orthogonal polynomials, together with previous related works~\cite{CMGV-1,CMGV-2}, allows to exploit various classical results~\cite{Chihara,Geronimus,Ismail,Khrushchev,NikiforovSuslovUvarov,Simon-OPUC,Szego} and new approaches to the theory of orthogonal polynomials, see for example~\cite{DTMVZ,VilenkinKlimyk,WvA}, to study related problems of quantum computation.

\appendix

\section{Orthogonal polynomials with periodic coefficients}
In this additional section we write down explicit formulas, in terms of Chebyshev polynomials of the second kind, for orthogonal polynomials with periodic coefficients. They generalize known expressions for OPUC with constant \cite{Golinskii-J} and two-periodic (see Section~\ref{sec:OPUC-2p}) Verblunsky coefficients. We provide also the corresponding formulas for OPRL which generalize known results given in~\cite{BeckermannGilewiczLeopold}. The direct tools of our proof can be extended to the case of eventually periodic coefficients.

We do not enter into other aspects of the theory, as the spectral representation of the OPUC with periodic Verblunsky coefficients is presented in detail in~\cite{Simon-OPUC}. The measure for OPRL whose recurrence coefficients are periodic is given in~\cite{GeronimoVanAssche}.

\subsection{OPUC with periodic Verblunsky coefficients}
Consider $p$-periodic Verblunsky coefficients $\alpha_{n+p} = \alpha_{n}$ and define the corresponding transfer matrix
\begin{equation*}
T_p(z) = A_{p-1}(z)  \dots  A_1(z) A_0(z), 
\end{equation*}
where we recall that
\begin{equation*} 
A_k(z) = \frac{1}{\rho_k}\left( \begin{array}{cc} z & -\bar\alpha_k \\ - \alpha_k z & 1 \end{array} \right), \qquad \rho_k = \sqrt{1 - |\alpha_k|^2}. 
\end{equation*}
Because $\det A_k(z) = z$ we infer that $\det T_p(z) = z^p$.
The discriminant 
\begin{equation*}
\Delta_p(z) = \frac{1}{z^{p/2}} \, \mathrm{Tr} [ T_p(z)].
\end{equation*} 
plays pivotal role (see for example Chapter~11 of monograph~\cite{Simon-OPUC}) in the theory of ``periodic" orthogonal polynomials. In particular, it is real-valued for $z\in\partial\DD$, and the absolutely continuous part of the spectral measure does not vanish if and only if $-2 < \Delta_{p}(e^{i\theta}) < 2$.

The following formula can be checked by simple calculation, but actually it follows form the Cayley--Hamilton theorem.
\begin{Lem}
	The discriminant and the transfer matrix for $p$-periodic Verblunsky coefficients satisfy equation
	\begin{equation} \label{eq:TpTpD}
	z^{-p/2} T_p(z) + z^{p/2} T_p^{-1}(z) = \Delta_p(z) \II_2. 
	\end{equation}	 
\end{Lem}

Let us give the explicit form of the corresponding orthogonal polynomials, in terms of the Chebyshev polynomials of the second kind 
\begin{equation} \label{eq:Ch-rec}
U_{k+1}(y) = 2y U_k(y) - U_{k-1}(y), \qquad U_{-1}(y) \equiv 0, \qquad U_0(y) \equiv 1 .
\end{equation}
\begin{Prop} \label{prop:p-OPUC}
	The orthonormal polynomials corresponding to $p$-periodic Verblunsky coefficients are given by
	\begin{equation*} 
	\begin{pmatrix} \varphi_{kp+j}(z) \\ \varphi_{kp+j}^*(z) \end{pmatrix} = z^{kp/2} \left[  U_k \left( \frac{\Delta_p(z)}{2}\right) A_{j-1}(z) \dots A_0(z) - z^{p/2}  U_{k-1} \left( \frac{\Delta_p(z)}{2} \right) A^{-1}_j(z) \dots A^{-1}_{p-1}(z) \right]
	\begin{pmatrix}  \varphi_{0}(z) \\ \varphi_{0}^*(z) \end{pmatrix},
	\end{equation*}
	where $ j=0, 1 , \dots , p-1$, and $k = 0, 1, 2, \dots $.
\end{Prop}
\begin{proof}
	We will concentrate on the case $j=0$, which reads
	\begin{equation*} 
	\begin{pmatrix} \varphi_{kp}(z) \\ \varphi_{kp}^*(z) \end{pmatrix} = z^{kp/2} \left[  U_k \left( \frac{\Delta_p(z)}{2}\right) \II_2 - z^{p/2}  U_{k-1} \left( \frac{\Delta_p(z)}{2} \right) T_p^{-1}(z) \right]
	\begin{pmatrix}  \varphi_{0}(z) \\ \varphi_{0}^*(z) \end{pmatrix},
	\end{equation*}
	because the final formula for other $j$ follows from
	\begin{equation*}
	\begin{pmatrix} \varphi_{kp+j}(z) \\ \varphi_{kp+j}^*(z) \end{pmatrix} = 
	A_{j-1}(z) \dots A_0(z)\begin{pmatrix} \varphi_{kp}(z) \\ \varphi_{kp}^*(z) \end{pmatrix}.
	\end{equation*}
	
	Let us proceed by mathematical induction with respect to $k$ (recall that we keep $j=0$). The starting point $k=0$ holds due to the form of the initial Chebyshev polynomials. The induction step follows from formula~\eqref{eq:TpTpD} and the recurrence \eqref{eq:Ch-rec}
	\begin{multline*} \qquad \qquad
	\begin{pmatrix} \varphi_{kp+p}(z) \\ \varphi_{kp+p}^*(z) \end{pmatrix} = z^{kp/2} \left[  U_k \left( \frac{\Delta_p(z)}{2}\right) T_p(z) - z^{p/2}  U_{k-1} \left( \frac{\Delta_p(z)}{2} \right)  \II_2 \right]
	\begin{pmatrix}  \varphi_{0}(z) \\ \varphi_{0}^*(z) \end{pmatrix} \stackrel{\eqref{eq:TpTpD}}{=} \\ 
	z^{(k+1)p/2} \left[ \left(  \Delta_p(z) U_k \left( \frac{\Delta_p(z)}{2}\right) - U_{k-1} \left( \frac{\Delta_p(z)}{2}\right) \right)\II_2  - z^{p/2}  U_{k} \left( \frac{\Delta_p(z)}{2} \right) T_p^{-1}(z) \right]
	\begin{pmatrix}  \varphi_{0}(z) \\ \varphi_{0}^*(z) \end{pmatrix} \stackrel{\eqref{eq:Ch-rec}}{=} \\
	z^{(k+1)p/2} \left[  U_{k+1} \left( \frac{\Delta_p(z)}{2}\right) \II_2  - z^{p/2}  U_{k} \left( \frac{\Delta_p(z)}{2} \right) T_p^{-1}(z) \right]
	\begin{pmatrix}  \varphi_{0}(z) \\ \varphi_{0}^*(z) \end{pmatrix} . \qquad \qquad \qquad
	\end{multline*}
\end{proof}
\begin{Rem}
	Notice that we didn't use the particular form of the initial polynomials $\varphi_{0}(z)$ and $\varphi_{0}^*(z)$. Our technique applies then, with appropriate shift in the final formula, also to orthogonal polynomials with eventually periodic (this nomenclature comes from the theory of continued fractions) Verblunsky coefficients, i.e. $\alpha_{n+p} = \alpha_n$ for $n\geq n_0$. In particular, as the initial polynomials we should take $\varphi_{n_0}(z)$ and $\varphi_{n_0}^*(z)$.
\end{Rem}

\subsection{OPRL with periodic recurrence relations} \label{sec:per-OPRL}
The technique used to express OPUC with periodic coefficients can be easily transferred to study the analogous problem for OPRL. It turns out that our final formula can be considered as a generalization and explanation of the main result of~\cite{BeckermannGilewiczLeopold}. We remark that our study does not concern Jacobi matrices with periodic coefficients~\cite{Moerbeke} whose theory involves more advanced algebro-geometric tools~\cite{BBEIM}. 

Following~\cite{DamanikKillipSimon} let us write the recurrence relation~\eqref{eq:3trr} in the matrix form
\begin{equation*}
\begin{pmatrix} Q_{k+1}(x) \\ Q_k(x) \end{pmatrix} =
\frac{1}{p_k} \begin{pmatrix} x - r_k & -q_k  \\ p_k & 0 \end{pmatrix}
\begin{pmatrix} Q_k(x) \\ Q_{k-1}(x) \end{pmatrix} = B_k(x) \begin{pmatrix} Q_k(x) \\ Q_{k-1}(x) \end{pmatrix}.
\end{equation*}
In the periodic case 
\begin{equation*}
p_{k+L} = p_k, \qquad q_{k+L} = q_k, \qquad r_{k+L} = r_k, \qquad k \geq 0.
\end{equation*}
define the corresponding transfer matrix
\begin{equation*}
T_L(x) = B_{L-1}(x) \dots B_1(x) B_0(x),
\end{equation*}
with
\begin{equation*}
\det T_L(x) = \frac{q_{(L)}}{p_{(L)}}, \qquad p_{(L)} =  p_0 p_1 \dots p_{L-1}, \qquad q_{(L)} =  q_0 q_1 \dots q_{L-1}.
\end{equation*}
The natural definition of the discriminant 
\begin{equation*}
\Delta_L(x) = p_{(L)} \mathrm{Tr}[T_L(x)], 
\end{equation*}
leads, by the Cayley--Hamilton theorem, to the identity
\begin{equation} \label{eq:TpTpD-R}
p_{(L)} T_L(x) + q_{(L)} T_L^{-1}(x) = \Delta_L(x) \II_2.
\end{equation}	 
The following result is direct analogue of Proposition~\ref{prop:p-OPUC} and its proof goes along the same lines as above, so we skip it.
\begin{Prop}
	The orthogonal polynomials corresponding to $L$-periodic three-term recurrence coefficients are given by
	\begin{equation} \label{eq:per-OPRL}
	\begin{split}
	\begin{pmatrix} Q_{kL+j}(x) \\ Q_{kL+j-1}(x) \end{pmatrix} & = \left(\frac{q_{(L)}}{p_{(L)}} \right)^{k/2} \left[  U_k \left( \frac{\Delta_L(x)}{2\sqrt{p_{(L)}q_{(L)}}} \right)  B_{j-1}(x) \dots B_0(x) \right. - \\ &
	\left.
	\left(\frac{q_{(L)}}{p_{(L)}} \right)^{1/2}  U_{k-1} \left( \frac{\Delta_L(x)}{2\sqrt{p_{(L)}q_{(L)}}} \right) B^{-1}_j(x) \dots B^{-1}_{L-1}(x) \right]
	\begin{pmatrix}  Q_0(x) \\ Q_{-1}(x) \end{pmatrix}, 
	\end{split}
	\end{equation}	
	where $ j=0, 1 , \dots , L-1$, and $k = 0, 1, 2, \dots $.
\end{Prop}
\begin{Rem}
Notice that in the proof we didn't use the particular form of the polynomials $Q_0(x)$ and $Q_{-1}(x)$. Our reasoning applies then also to orthogonal polynomials with eventually periodic coefficients 
\begin{equation*}
p_{k+L} = p_k, \qquad q_{k+L} = q_k, \qquad r_{k+L} = r_k, \qquad k \geq k_0.
\end{equation*}
One should appropriately shift by $k_0$ the final formula \eqref{eq:per-OPRL} and modify the initial polynomials to $Q_{k_0}(x)$ and $Q_{k_0 -1}(x)$.
\end{Rem}
\begin{Rem}
	With $k_0 = 1$ and $L=1$ we recover formula~\eqref{eq:OPRL-1p}.
\end{Rem}

\bibliographystyle{amsplain}

\providecommand{\bysame}{\leavevmode\hbox to3em{\hrulefill}\thinspace}

\end{document}